\documentclass[reprint,prr,amsmath,amssymb, twocolumn,superscriptaddress]{revtex4-2}

\usepackage{graphicx}
\usepackage{dcolumn}   % needed for some tables
\usepackage{bm}        % for math
\usepackage{bbm}
\usepackage[colorlinks=true, urlcolor=blue,citecolor=blue,anchorcolor=blue]{hyperref}
\usepackage[capitalise]{cleveref}
\usepackage{dsfont}
\usepackage{mathtools}
\usepackage{amsthm, thmtools}
\usepackage{amsmath}
\usepackage{empheq}
\newcommand{\sgn}{\mathrm{sgn}}
\usepackage{scalerel}
\usepackage{dsfont}

\newcommand{\tinyspace}{\mspace{1mu}}
\newcommand{\microspace}{\mspace{0.5mu}}

\newcommand{\norm}[1]{\left\lVert\tinyspace#1\tinyspace\right\rVert}

\declaretheorem[
]{lemma}

\def\<{\langle}
\def\>{\rangle}
\def \lket {\left|}
\def \rket {\right\rangle}
\def \lbra {\left\langle}
\def \rbra {\right|}
\newcommand{\ket}[1]{\lket\microspace #1 \microspace\rket}
\newcommand{\bra}[1]{\lbra\microspace #1 \microspace\rbra}

\usepackage{mathtools}

\usepackage{bm}
\let\vec\bm

\newcommand{\beq}{\begin{equation}}
\newcommand{\eeq}{\end{equation}}

\begin{document}
\title{Optimal function estimation with photonic quantum sensor networks}
\date{\today}
\author{Jacob Bringewatt}
\thanks{These authors contributed equally.}
\affiliation{Joint Center for Quantum Information and Computer Science, NIST/University of Maryland College Park, Maryland 20742, USA}
\affiliation{Joint Quantum Institute, NIST/University of Maryland College Park, Maryland 20742, USA}
\author{Adam Ehrenberg}
\thanks{These authors contributed equally.}
\affiliation{Joint Center for Quantum Information and Computer Science, NIST/University of Maryland College Park, Maryland 20742, USA}
\affiliation{Joint Quantum Institute, NIST/University of Maryland College Park, Maryland 20742, USA}
\author{Tarushii Goel}
\thanks{These authors contributed equally.}
\affiliation{Joint Center for Quantum Information and Computer Science, NIST/University of Maryland College Park, Maryland 20742, USA}
\affiliation{Joint Quantum Institute, NIST/University of Maryland College Park, Maryland 20742, USA}
\affiliation{Department of Physics, Massachusetts Institute of Technology, Cambridge, Massachusetts 02139, USA}
\author{Alexey V. Gorshkov}
\affiliation{Joint Center for Quantum Information and Computer Science, NIST/University of Maryland College Park, Maryland 20742, USA}
\affiliation{Joint Quantum Institute, NIST/University of Maryland College Park, Maryland 20742, USA}
\begin{abstract}
    The problem of optimally measuring an analytic function of unknown local parameters each linearly coupled to a qubit sensor is well understood, with applications ranging from field interpolation to noise characterization. Here, we resolve a number of open questions that arise when extending this framework to Mach-Zehnder interferometers and quadrature displacement sensing. In particular, we derive lower bounds on the achievable mean square error in estimating a linear function of either local phase shifts or quadrature displacements. In the case of local phase shifts, these results prove, and somewhat generalize, a conjecture by Proctor \emph{et al.} [arXiv:1702.04271 (2017)]. For quadrature displacements, we extend proofs of lower bounds to the case of arbitrary linear functions. We provide optimal protocols achieving these bounds up to small (multiplicative) constants and describe an algebraic approach to deriving new optimal protocols, possibly subject to additional constraints. Using this approach, we prove necessary conditions for the amount of entanglement needed for any optimal protocol for both local phase and displacement sensing.
 \end{abstract}
\maketitle

\section{Introduction}
In quantum metrology, entangled states of quantum sensors are used to try to obtain a performance advantage in estimating an unknown parameter or parameters (e.g.,~field amplitudes) coupled to the sensors. In addition to this practical advantage of quantum sensing, the theory of the ultimate performance limits for parameter estimation tasks is deeply related to a number of topics of theoretical interest in quantum information science, such as resource theories~\cite{piani2016robustness}, the geometry of quantum state space~\cite{braunstein1994statistical}, quantum speed limits~\cite{taddei2013quantum,deffner2017quantum,garciapintos2022unifying}, and quantum control theory~\cite{deffner2017quantum}. 

Initial experimental and theoretical work on quantum sensing focused on optimizing the estimation of a single unknown parameter (see, e.g., Ref.~\cite{pezze2018quantum} for a review). More recently, the problem of distributed quantum sensing has become an area of particular interest~\cite{zhang2021distributed}. Here, one considers a network of quantum sensors, each coupled to a local unknown parameter. The prototypical task in this setting is to measure some function or functions of these parameters. 
In this context, the task of optimally measuring a single linear function $q(\vec\theta)$ of $d$ independent local parameters $\vec\theta=(\theta_1,\cdots,\theta_d)^T$ is particularly well studied both theoretically~\cite{proctor2017networked,ge2018distributed,eldredge2018optimal,proctor2018multiparameter,zhuang2018distributed,xia2019repeater,gross2020one,rubio2020quantum,oh2020optimal,triggiani2021heisenberg, ehrenberg2023minimum,oh2022distributed,malitesta2021distributed,yang2023quantum} and experimentally~\cite{guo2020distributed,xia2020demonstration, liu2021distributed,zhao2021field_etal}. In addition to its independent utility (i.e.,~for measuring an average of local fields in some region), linear function estimation serves as a key subtask of more general metrological tasks, such as measuring an analytic function of the unknown parameters~\cite{qian2019heisenberg}, measuring an analytic function of dependent parameters~\cite{qian2021optimal, hamann2022approximate}, or measuring multiple functions~\cite{rubio2020bayesian, bringewatt2021protocols}.

For qubit sensors, the asymptotic limits on performance for these function estimation tasks are rigorously understood, and techniques for generating optimal protocols subject to various constraints, such as limited entanglement between sensors, are known~\cite{ehrenberg2023minimum}. However, despite extensive theoretical and experimental research on distributed quantum sensing for photonic quantum sensors (see, e.g.,~\cite{polino2020photonic,zhang2021distributed} for reviews), the asymptotic performance limits for function estimation are not yet rigorously established. Here, we close this gap, proving an ultimate bound on asymptotic performance, as measured by the mean square error of the estimator, for measuring a linear function of unknown parameters each coupled to a different photonic mode via either (1) the number operator $\hat n$ or (2) a field-quadrature operator, chosen without loss of generality to be the momentum quadrature $\hat{p}:=i(\hat{a}^\dagger-\hat{a})/2$. That is, we are interested in determining a function of either unknown local phase shifts or unknown quadrature displacements. For case (1), our primary focus, we derive this bound subject to a strict constraint on photon number, proving a long-standing conjecture appearing in Ref.~\cite{proctor2017networked}. In case (2), we derive our bound subject to a constraint on the average photon number, which is more natural in this setting as quadrature displacements are not photon-number conserving. Here, our results are consistent with existing bounds in the literature~\cite{xia2019repeater}, but, for completeness, we include derivations in this setting using an equivalent mathematical framework to the number operator case and the qubit sensor case~\cite{ehrenberg2023minimum}. This allows for a natural comparison of the various performance limits and resource requirements of function estimation in quantum sensor networks and opens the door to designing new, information-theoretically optimal protocols in the asymptotic limit of sufficient data.

The rest of the paper proceeds as follows. In Sec.~\ref{sec:setup}, we formally set up the problem of interest and provide useful notation. In Sec.~\ref{sec:bounds} we prove lower bounds on the mean-squared error of an estimator for arbitrary linear functions for both number operator and displacement operator generators. We then study protocols that saturate these bounds in Sec.~\ref{sec:protocols}. Finally, we discuss other entanglement-restricted optimal protocols in Sec.~\ref{sec:entanglement}. 

\section{Problem setup}\label{sec:setup}
Consider a sensor network of $d$ optical modes each coupled to an unknown parameter $\theta_j$ for $j\in\{1, \cdots, d\}$ via  
\begin{equation}\label{eq:initialH}
    \hat{H}(s) = \sum_{j=1}^d\theta_j \hat g_j +\hat{H}_c(s) =: \vec\theta\cdot\hat{\vec g}+\hat{H}_c(s),
\end{equation}
where $\hat g_j$ is the local coupling Hamiltonian and boldface denotes vectors. Here, we consider the following two cases:
\begin{subequations}
\begin{align}
    \hat g_j&:=\hat n_j = \hat a^\dagger_j \hat a_j, \label{eq:particle_number_coupling}\\
    \hat g_j&:=\hat p_j = \frac{i}{2}(\hat a^\dagger_j-\hat a_j), \label{eq:displacement_coupling}
\end{align}
\end{subequations}
where $\hat a^\dagger_j, \hat a_j$ are the bosonic creation and annihilation operators acting on mode $j$, $\hat{n}_j$ is the number operator acting on mode $j$, and $\hat p_j$ is the momentum- ($\hat{p}$-) quadrature on mode $j$. The choice of $\hat{p}$ quadrature is, of course, arbitrary. All results apply equally well for coupling to any quadrature. The $\vec\theta$-independent, time-dependent Hamiltonian $\hat H_c(s)$ is a control Hamiltonian, possibly including coupling to an arbitrary number of ancilla modes. Here, $s\in[0,t]$, where $t$ is the total sensing time.

In either case, our task is to measure a linear function $q(\vec\theta)=\vec\alpha\cdot\vec\theta$ of the local field amplitudes $\vec\theta$ where $\vec\alpha\in\mathbb{Q}^d$ is a vector of rational coefficients. (The restriction to rational coefficients is due to the discreteness of the resources---the number of photons---available in this problem; in the case we are interested in---large photon numbers---this is only a technical point.) To accomplish this task, we consider probe states with either fixed photon number $N$ or fixed average photon number $\overline{N}$. Given such probe states, we consider encoding the unknown parameters into the state via the unitary evolution generated by the Hamiltonian in Eq.~(\ref{eq:initialH}).

We will consider both an unrestricted control Hamiltonian and a control Hamiltonian fixed to have the form
\begin{equation}\label{eq:controlHnumber}
\hat H_c(s)=\hat{h}_c(s)\delta(s-j\Delta t),
\end{equation}
where $\hat{h}_c(s)$ is a (unitless) Hermitian operator, $\delta(s)$ is the Dirac delta function, $\Delta t:=t/M$ is the time for a single application of the encoding unitary $\exp(-iH\Delta t)$. The index $j\in \{1, \cdots, M\}$ indexes these applications, where $M$ is the total number of applications. This construction is motivated by the fact that typical physical implementations of a number operator coupling, e.g.,~in a Mach-Zehnder interferometer, and displacement operator coupling, e.g.,~via an electro-optical modulator (EOM), often do not allow for intermediate controls at arbitrary times. Therefore, when we fix our control Hamiltonians to be described by Eq.~(\ref{eq:controlHnumber}), we have limited any controls to be applied between each pass through these optical elements; for simplicity, we have assumed that these control operations can be implemented on a timescale much shorter than the timescale of phase accumulation. Without loss of generality, we will let $\Delta t=1$ for the rest of this paper, implying that (in this setting) $t=M$. Therefore, the parameter encoding procedure for the photon number coupling is done via the unitary
\begin{equation}\label{eq:full_unitary}
U=U^{(M)}VU^{(M-1)}V\cdots U^{(1)}V=\prod_{m=1}^M (U^{(m)}V),
\end{equation}
where $V:=\exp(-i\hat{\vec{g}}\cdot\vec\theta)$ and $U^{(m)}$ for $m\in\{1,\cdots, M\}$ denote the unitaries applied between passes. Here, by pass, we mean a single application of the unitary $V$. We use the convention that the product operation left multiplies.

In both settings, it is worth emphasizing that, while our information-theoretic results lower bounding the asymptotically achievable mean square error of an estimate $\tilde{q}$ of $q$ will apply to any protocol within the framework(s) described above, the explicit protocols we will develop will use finite ancillary modes and finite controls.

\section{Lower bounds}\label{sec:bounds}
Following the approach of Refs.~\cite{eldredge2018optimal,ehrenberg2023minimum}, we compute lower bounds on the mean square error $\mathcal{M}$ of an estimator $\tilde{q}$ of $q$ 
by rewriting the Hamiltonian in Eq.~(\ref{eq:initialH}) as 
\begin{equation}\label{eq:Hnewbasis}
    H(s)=\sum_{j=1}^d(\vec\alpha^{(j)}\cdot\vec\theta)(\vec\beta^{(j)}\cdot\hat{\vec{g}}) + \hat{H}_{c}(s),
\end{equation}
for some (time-independent) choice of basis vectors $\{\vec\alpha^{(j)}\}_{j=1}^d$, where $\vec\alpha^{(1)}:=\vec\alpha$ and $\{\vec\beta^{(j)}\}_{j=1}^d$ is a dual basis such that $\vec\alpha^{(i)}\cdot\vec\beta^{(j)}=\delta_{ij}$. 
The vectors $\{\vec\alpha^{(j)}\}_{j=1}^d$ are associated with a change of basis $\vec\theta\rightarrow\vec{q}$ where $q_j:=\vec\alpha^{(j)}\cdot\vec\theta$ such that $q_1=q$; that is, $\vec\alpha^{(1)}=:\vec\alpha$ with corresponding dual vector $\vec\beta^{(1)}=:\vec\beta$. Then we can define a $\vec\beta$-parameterized generator of translations with respect to the quantity $q$ as
\begin{equation}\label{eq:defgq}
    \hat{g}_{q,\vec\beta}:=\min_{q^{(2)}, \cdots, q^{(d)}}\frac{\partial \hat{H}}{\partial q}\Bigg|_{q^{(2)}, \cdots, q^{(d)}}=\vec\beta\cdot\hat{\vec g}.
\end{equation}

Armed with Eq.~(\ref{eq:defgq}), we can write a bound on $\mathcal{M}$ in terms of a single-parameter quantum Cram\'{e}r-Rao bound~\cite{giovannetti2006quantum, boixo2007generalized,polino2020photonic}
\begin{equation}\label{eq:qcrb}
    \mathcal{M}\geq \frac{1}{\mu\mathcal{F}(q|\vec\beta)},
\end{equation}
where $\mathcal{F}(q|\vec\beta)$ is the quantum Fisher information with respect to $q$, given some choice of fixing the extra $d-1$ degrees of freedom in our problem, as specified by the vector $\vec\beta\in\mathbb{R}^d$ such that $\vec\alpha\cdot\vec\beta=1$. Any such single-parameter bound is a valid lower bound as fixing extra degrees of freedom can only give us more information about the parameter $q$ (see below for mathematical details). $\mu$ is the number of experimental repetitions. This bound holds for an unbiased estimator $\tilde q$. When deriving our bounds, we will restrict ourselves to single-shot Fisher information and set $\mu=1$~\footnote{Clearly, with $\mu=1$, we are not guaranteed the existence of an unbiased estimator, so there is some subtlety in this restriction. The choice is sufficient for determining bounds and optimal probe states, but, when considering measurements to extract the quantity of interest, realistic protocols must use more than one shot. For instance, robust phase estimation allows for $\mu=\mathcal{O}(1)$, while still allowing us to obtain an unbiased estimator that achieves the quantum Cram\'{e}r-Rao bound up to a multiplicative constant~\cite{kimmel2015robust,kimmel2015robusterratum,belliardo2020achieving}. In Appendix~\ref{app:robust-phase-estimation}, for completeness, we briefly summarize this approach. See also, Refs.~\cite{yang2019attaining,suzuki2020quantum} and Ref.~\cite{ehrenberg2023minimum} for further discussion of these issues.}. Quantum Fisher information is maximized for pure states, so restricting ourselves to pure states and unitary encoding of the unknown parameters into the state we can write
\begin{equation}\label{eq:qcrb2}
    \mathcal{F}(q|\vec\beta)\leq 4 t^2 \max_\rho [(\Delta\hat g_{q,\vec\beta})_\rho]^2,
\end{equation}
where $\hat g_{q,\vec\beta}$ is the $\vec\beta$-parameterized generator of translations with respect to the unknown function $q$. The variance $[\Delta(\hat g_{q,\vec\beta})_\rho]^2$ is taken with respect to a pure probe state $\rho=\ket{\psi}\bra{\psi}$.

Ultimately, we seek a choice of new basis that yields the tightest possible bound on the quantum Fisher information $\mathcal{F}(q)$. This choice is determined by the solution to~\footnote{Note the use of a minimax as opposed to a maximin in Eq.~(\ref{eq:minprob}). This follows from the fact that the minimax of some objective function is always greater than or equal to the maximin and we seek to maximize the quantum Fisher information.} 
\begin{align}\label{eq:minprob}
    \min_{\vec\beta} \max_\rho [\Delta(\vec\beta\cdot\hat{\vec{g}})_\rho]^2, \quad
    \text{subject to }\, \vec\alpha\cdot\vec\beta=1.
\end{align}
Let $(\vec\beta^*,\rho^*)$ be a solution for this optimization problem. Then we can rewrite the single-shot version of Eq.~(\ref{eq:qcrb}) as
\begin{equation}\label{eq:bnd_betastar}
    \mathcal{M}\geq \frac{1}{4 t^2[\Delta(\vec\beta^*\cdot\hat{\vec g})_{\rho^*}]^2}. 
\end{equation}
This bound can be understood as corresponding to the optimal choice of an imaginary single parameter scenario, where we have fixed $d-1$ of the $d$ parameters controlling the evolution of the state, leaving only the parameter of interest $q$ free to vary. While this requires giving ourselves information that we do not have, additional information can only reduce $\mathcal{M}$, and, therefore, any such choice provides a lower bound on $\mathcal{M}$ (via single-parameter bounds) when we do not have such information. 
While not guaranteed by this method of derivation, we shall see that such bounds are saturable, up to small multiplicative constants.

Constraints can be placed on the probe state $\rho$ depending on the physical generators coupled to the parameters of interest: as previously discussed, in this work we consider the constraints of fixed photon number $N$ for the generator $\hat{n}_j$ and fixed average photon number $\overline{N}$ for the generator $\hat{p}_j$. The rationale behind these constraints is as follows. $\hat{p}$ does not conserve photon number, hence it does not make sense to restrict to a fixed photon number sector when coupling to quadrature operators, and, thus, average photon number is the natural constraint. For $\hat{n}$, on the other hand, we must work in the fixed photon sector, as using fixed average photon number allows for the construction of pathological probe states enabling arbitrarily precise sensing. In particular, consider the state 
\begin{equation}
    \ket{\psi_{a}} = \sqrt{\frac{a-1}{a}}\ket{0} + \sqrt{\frac{1}{a}}\ket{a \overline{N}}.
\end{equation}
It is easy to see that $\ket{\psi_{a}}$ has mean photon number $\overline{N}$ and variance $(a-1)\overline{N}^{2}$. Hence, even for fixed $\overline{N}$, letting $a$ get arbitrarily large allows for an arbitrarily large variance, and hence arbitrarily precise sensing.

Leaving the details of the calculation to Appendix~\ref{app:bnd-phase}, solving the above optimization problem for $\hat g_j=\hat n_j$ restricted to probe states with exactly $N$ photons yields
\begin{align}\label{eq:finalbnd}
    \mathcal{M}&\geq 
    \frac{\max\left\{\norm{\vec\alpha}_{1,\mathcal{P}}^2, \norm{\vec\alpha}_{1,\mathcal{N}}^2\right\}}{N^2t^2},
\end{align}
where $\mathcal{P}:=\{j \,|\, \alpha_j\geq 0\}$ and $\mathcal{N}:=\{j \,|\, \alpha_j< 0\}$. In the second line, we use the notation
\begin{align}
    \norm{\vec\alpha}_{1,\mathcal{S}}&:=\sum_{i \in \mathcal{S}} |\alpha_i|, 
\end{align}
where $\mathcal{S}\in\{\mathcal{P},\mathcal{N}\}$. For the rest of the paper, we assume without loss of generality that we are in the case that $\norm{\vec\alpha}_{1,\mathcal{P}} \geq \norm{\vec\alpha}_{1,\mathcal{N}}$ to simplify our expressions. In the special case where $\vec{\alpha}$ possesses only positive coefficients (i.e., $\mathcal{N}=\emptyset$), 
\begin{equation}\label{eq:finalbnd_positivealpha}
    \mathcal{M}\geq \frac{\norm{\vec\alpha}_1^2}{N^2t^2},
\end{equation}
proving a long-standing conjecture from Ref.~\cite{proctor2017networked} that this is the minimum attainable variance for $\vec\alpha\in\mathbb{Q}^d$ with $\vec\alpha\geq 0$ and $N\vec{\alpha}\in\mathbb{N}^{d}$. This is our primary result. 

Similarly, for the case of local quadrature displacements restricted to probe states with average photon number $\overline{N}$, we obtain the following bound:
\begin{equation}\label{eq:dispbnd}
    \mathcal{M}\geq\frac{\norm{\vec\alpha}_2^2}{4\overline{N}t^2}-\mathcal{O}\left(\frac{d\norm{\vec\alpha}_2^2}{\overline{N}^2t^2}\right).
\end{equation}
Equation~(\ref{eq:dispbnd}) is a minor generalization of the results in Refs.~\cite{zhang2021distributed,xia2019repeater}, extended to allow for negative coefficients and for arbitrary non-Gaussian probe states. Therefore, for completeness, we include a reminder of the arguments from Refs.~\cite{zhang2021distributed,xia2019repeater} along with our more general derivation in Appendix~\ref{app:bnd-displacement}.

We can compare the bounds in Eqs.~(\ref{eq:finalbnd}) and (\ref{eq:dispbnd}) to the corresponding bounds on the mean square error obtainable by separable protocols---that is, those using separable probe states such that each parameter $\theta_i$ is measured individually using an optimized partition of the available photons, and then these estimates are used to compute $q$. In particular, for number operator coupling and fixed photon number states, using $\eta_j=\frac{|\alpha_j'|}{\norm{\vec\alpha'}_1}N$ photons ($\alpha'_j:=\alpha_j^{2/3}$) in mode $j$, it holds that~\cite{proctor2017networked}
\begin{equation}\label{eq:sepbnd}
    \mathcal{M}_\mathrm{sep}\geq \frac{\norm{\vec\alpha'}_{2/3}^2}{N^2t^2},
\end{equation}
where $\norm{\cdot}_{2/3}$ denotes the Schatten $p$-function
\begin{equation}
    \norm{\vec{v}}_{p} = \left(\sum_{i}v_{i}^{p}\right)^{1/p}
\end{equation}
with $p=2/3$. When $p \in [1,\infty]$, this function is a norm, but for $p \in (0,1)$ it is not, as it does not satisfy the property of absolute homogeneity, but it still provides a convenient notational shorthand.  

Performing a similar optimization for the case of displacement coupling and fixed average photon number, one obtains 
\begin{equation}\label{eq:sepbnddisp}
    \mathcal{M}_{\mathrm{sep}}\geq \frac{\norm{\vec\alpha}_1^2}{4\overline{N}t^2}+\mathcal{O}\left(\frac{1}{\overline{N}^2t^2}\right),
\end{equation}
where the optimum division of photons is given by using $\eta_j=\frac{|\alpha_j|}{\norm{\vec\alpha}_1}N$ photons in mode $j$. A non-closed-form version of this bound can be found in Ref.~\cite{zhuang2018distributed} in the case where $\overline{N}$ is finite. One recovers our result in the asymptotic in $\overline N$ limit.

Consequently, in both the phase and displacement sensing settings, the achievable advantage due to entanglement between modes is fully characterized by the difference between the vector $p$ norm of $\vec\alpha$ with $p=\frac{2}{3},1$ or $p=1,2$, respectively. By generalized H{\"o}lder's inequality, $\norm{\vec\alpha}_{2/3}^2\leq d\norm{\vec\alpha}_{1}^2$
and $\norm{\vec\alpha}_1^2\leq d\norm{\vec\alpha}_2^2$. 
Both inequalities are saturated for any ``average-like'' function with $|\vec\alpha|\propto(1,1,\cdots 1)^T$. In both cases, we obtain a $\mathcal{O}(1/d)$ improvement in precision due to entanglement, consistent with the so-called Heisenberg scaling in the number of sensors $d$. This is consistent with results for qubits in Ref.~\cite{eldredge2018optimal}, where the best improvement between the separable and entangled bounds occurs when measuring an average-like function. For the case of phase sensing, the optimal performance, including constants, is obtained when $\norm{\vec\alpha}_{1,\mathcal{P}}^2=\norm{\vec\alpha}_{1,\mathcal{N}}^2=\norm{\vec\alpha}_1/2$ (which occurs when the vector $\vec\alpha$ is half positive ones and half negative ones).
 
\section{Protocols}\label{sec:protocols}
\subsection{Existing protocols}
The bounds established in the previous section are all saturable, up to small multiplicative constants, using protocols that exist in the literature, or slight variations thereof. In particular, Refs.~\cite{proctor2017networked, proctor2018multiparameter} present a protocol for estimating a linear function of local phase shifts with positive coefficients (i.e., $\vec\alpha\geq 0$) which achieves the bound in Eq.~(\ref{eq:finalbnd}) up to a small multiplicative constant. This protocol makes use of a so-called proportionally weighted N00N state over $d+1$ modes, 
\begin{align}\label{eq:noon}
    \ket{\psi}\propto \bigg|N\frac{\alpha_1}{\norm{\vec\alpha}_1}, \cdots, N\frac{\alpha_d}{\norm{\vec\alpha}_1}, 0\bigg \rangle+\bigg|0,\cdots, 0,N\bigg\rangle,
\end{align}
where we have expressed the state in an occupation number basis over $d+1$ modes and have dropped the normalization for concision. The last mode serves as a reference mode. Observe that, for this state to be well defined, it is essential that $\vec\alpha/\norm{\vec{\alpha}}_{1}\in\mathbb{Q}^d$ and that $N$ is sufficiently large that the resulting occupation numbers are integers. Details of how protocols using this probe state work and how they generalize to the case of negative coefficients are provided in Appendix~\ref{app:protocols}. A description of how to achieve the separable bound in Eq.~(\ref{eq:sepbnd}) is provided in Appendix~\ref{app:bnd-displacement}.

Similarly, in the case of measuring a linear function of displacements using states with fixed average photon number, Ref.~\cite{zhuang2018distributed} provides a protocol that, up to small multiplicative constants, saturates the bound in Eq.~(\ref{eq:dispbnd}), and a separable protocol that, again up to small constants, achieves the bound in Eq.~(\ref{eq:sepbnddisp}). Interestingly, these protocols require only Gaussian probe states, indicating that these states are optimal. In particular, these protocols make use of an initial single-mode squeezed state, followed by a properly constructed beam-splitter array to prepare a multimode entangled probe state with the appropriate sensitivity to quadrature displacements in each mode. Homodyne measurements on each mode can then be used to extract the function of interest. Consistent with this fact, our separable lower bound matches the Gaussian state-restricted bound obtained in Ref.~\cite{zhuang2018distributed} and the bound for arbitrary states derived in Ref.~\cite{xia2019repeater} for the particular case of measuring an average.

\subsection{Algebraic conditions for new protocols}\label{s:new-prot} 
Other protocols are possible and can be derived via a simple set of algebraic conditions. In particular, for a probe state to exist saturating the bound in Eq.~(\ref{eq:bnd_betastar}), or its specific versions in Eqs.~(\ref{eq:finalbnd}) and (\ref{eq:dispbnd}), we require the existence of an optimal choice of basis transformation $\vec\theta\rightarrow\vec q$ such that knowing $q_j$ for $j>1$ yields no information about $q=q_1$. Mathematically, this means that the quantum Fisher information matrix~\cite{liu2019quantum} with respect to the parameters $\vec q$ must have the following properties:
\begin{subequations}
\begin{align}\label{eq:condonfim}
    \mathcal{F}(\vec q)_{11}&=4 t^2 [\Delta(\vec\beta^*\cdot{\vec{\hat g}})_{\rho^*}]^2,\\
    \mathcal{F}(\vec q)_{1i}&=\mathcal{F}(\vec q)_{i1}=0\quad (\forall\, i\neq 1),\label{eq:condonfimoffdiagonal}
\end{align}
\end{subequations}
Recall that $(\vec\beta^*,\rho^*)$ are the solution to the minimax problem in Eq.~(\ref{eq:minprob}).
We can reexpress these conditions in terms of the quantum Fisher information matrix with respect to $\vec\theta$ as
\begin{subequations}
\begin{align}\label{eq:condonfim_theta}
    (\vec\beta^*)^T\mathcal{F}(\vec \theta)\vec\beta^*&=4 t^2 [\Delta(\vec\beta^*\cdot{\vec{\hat g}})_{\rho^*}]^2,\\
    (\vec\beta^{*})^T\mathcal{F}(\vec \theta)\vec\beta^{(i)}&=(\vec\beta^{(i)})^T\mathcal{F}(\vec \theta)\vec\beta^*=0\quad (\forall\, i\neq 1) \label{eq:condonfimoffdiagonal_theta}.
\end{align}
\end{subequations}
Then, using $\vec\alpha^{(i)}\cdot\vec\beta^{(j)}=\delta_{ij}$, we obtain the condition
\begin{equation}\label{eq:finalcondsat}
    \mathcal{F}(\vec\theta)\vec\beta^*=4 t^2[\Delta(\vec\beta^*\cdot{\vec{\hat g}})_{\rho^*}]^2\vec\alpha.
\end{equation}

Matrix elements of $\mathcal{F}(\vec\theta)$ for pure probe states and unitary evolution are given via
\begin{align}\label{eq:fim}
    \mathcal{F}(\vec\theta)_{ij}&=4\left[\frac{1}{2}\langle\{\mathcal{H}_i,\mathcal{H}_j\}\rangle-\langle\mathcal{H}_i\rangle\langle\mathcal{H}_j\rangle\right],
\end{align}
where $\mathcal{H}_i=-iU^\dagger\partial_i U$ with $\partial_i:=\partial/\partial\theta_i$, $U$ is the unitary generated by Eq.~(\ref{eq:initialH}) and the expectation values are taken with respect to the initial probe state~\cite{liu2019quantum}. 

We refer to protocols that make use of probe states and controls so that Eq.~(\ref{eq:finalcondsat}) is satisfied as optimal. However, we caution that the existence of an optimal probe state does not imply the existence of measurements on this state that allow one to extract an estimate of the parameter $q$ saturating the lower bounds we have derived.  This issue of the optimal measurements to extract parameters is also discussed extensively in, e.g.~Ref.,~\cite{hayashi2018resolving}, with some convenient, nearly optimal, protocols presented in Refs.~\cite{kimmel2015robust,kimmel2015robusterratum,belliardo2020achieving}. Such methods are the origin of the ``small multiplicative constants'' that arise in the explicit protocols above. In fact, lower bounds derived via the quantum Cram\'{e}r-Rao bound can be obtained only up to a constant $\geq \pi^2$~\cite{gorecki2020pi}. See Appendix~\ref{app:robust-phase-estimation} for a brief explanation of these ideas.

For the particular cases considered in this paper, $\vec\beta^*$ has been explicitly calculated (see Appendices~\ref{app:bnd-phase} and \ref{app:bnd-displacement}), so Eq.~(\ref{eq:finalcondsat}) can be expressed in a more meaningful form. For number operator coupling, we obtain the condition
\begin{equation}\label{eq:finalcond}
    \sum_{i\in\mathcal{P}}\mathcal{F}(\vec\theta)_{ij}=\frac{N^2t^2}{\norm{\vec\alpha}_{1,\mathcal{P}}}\alpha_j,
\end{equation}
for all $j$. Similarly, for the quadrature coupling, an optimal protocol requires 
\begin{equation}\label{eq:finalcondquad}
    \mathcal{F}(\vec\theta)\vec\alpha\sim4\overline{N}t^2\vec\alpha,
\end{equation}
where $\sim$ denotes asymptotically in $\overline{N}$. Equations~(\ref{eq:finalcond}) and (\ref{eq:finalcondquad}) provide a generic route to finding new protocols: consider a set of parameterized families of probe states $\mathcal{T}$ that one can coherently switch between using available controls $\hat H_c(t)$ (here, a ``family'' of states refers to a particular superposition of Fock states with an arbitrary relative phase). One can then calculate $\mathcal{F}(\vec\theta)$ via Eq.~(\ref{eq:fim}) and allocate the time spent in a particular family of states such that the associated quantum Fisher information condition is achieved. As a limiting case, one could consider $|\mathcal{T}|=1$, removing the necessity of coherent control; the protocols considered in the previous section are of this sort (and, in Appendix~\ref{app:protocols}, we show that these protocols do, indeed, achieve the saturability conditions). 

The possible choices for families of states $\mathcal{T}$ that allow for such a solution are actually quite limited, even given access to arbitrary control Hamiltonians and ancilla modes. In particular, we prove the following in the case where $\hat g_j:=\hat n_j$:

\begin{lemma}\label{lemma:opt_states}
Any optimal protocol using $N$ photons and $M$ passes through interferometers with a coupling as in Eq.~(\ref{eq:initialH}) with $\hat{g}_j=\hat{n}_j$ requires that, for every pass $m$, the probe state $\ket{\psi_m}$ be of the form
\begin{equation}
    \ket{\psi_m} \propto \ket{\vec{N}(m)}_{\mathcal{P}}\ket{\vec{0}}_{\mathcal{NR}} + e^{i\varphi_m}\ket{\vec{0}}_{\mathcal{P}}\ket{\vec{N'}(m)}_{\mathcal{NR}},
\end{equation}
where $\mathcal{P}$, $\mathcal{N}$, and $\mathcal{R}$ represent the modes with $\alpha_j\geq 0$, $\alpha_j<0$, and the (arbitrary number of) reference modes, respectively, $\vec{N}(m)$ and $\vec{N'}(m)$ are strings of occupation numbers such that $|\vec{N}(m)|=|\vec{N'}(m)|=N$ for all passes $m$. $\varphi_m$ is an arbitrary phase. 
\end{lemma}
The proof follows straightforwardly from an explicit calculation of the Fisher information matrix for $\hat{g}_j=\hat{n}_j$, but is somewhat algebraically tedious so we relegate it to Appendix~\ref{app:proof-of-lemma-1}.

Lemma~\ref{lemma:opt_states} suggests a particular choice of $\mathcal{T}$ from which we can pick an optimal protocol for function estimation in the $\hat{g}_j=\hat{n}_j$ case. In particular, define a set of vectors
\begin{align}\label{eq:omegadef}
    \mathcal{W} := &\left\{\vec\omega\in\mathbb{Z}^d\, \big|\, \norm{\vec\omega}_{1,\mathcal{P}}=N,\, \norm{\vec\omega}_{1,\mathcal{N}}\leq N,\, \omega_j\alpha_j\geq 0 \, \forall\, j\right\}.
\end{align}
Further, consider the restriction $\vec\omega|_\mathcal{P}\in\mathbb{Z}^d$ with components
\begin{align}\label{eq:omega_restriction}
   (\vec\omega|_\mathcal{P})_j=\begin{cases}
   \omega_j, &j \in\mathcal{P}\\
   0, &\text{otherwise},
   \end{cases}
\end{align}
and the restriction $\vec\omega|_\mathcal{N}$, defined similarly. Armed with these vectors, we can define a particular choice $\mathcal{T}$ of one-parameter families of probe states in an occupation number basis where each $\ket{\psi(\vec\omega;\varphi)}\in \mathcal{T}$ is labeled by a particular choice of $\vec\omega$ such that
\begin{align}\label{eq:omegastates}
    &\ket{\psi(\vec\omega;\varphi)}\propto \ket{\vec \omega|_\mathcal{P}}\ket{0}+e^{i\varphi}\ket{-\vec\omega|_{\mathcal{N}}}\ket{N-\norm{\vec\omega|_\mathcal{N}}_1},
\end{align}
where $\varphi\in\mathbb{R}$ is an arbitrary parameter and the last mode is a reference mode. It should be clear that these families of states are of the form specified by Lemma~\ref{lemma:opt_states}. Furthermore, note that the proportionally weighted N00N state in \cref{eq:noon} is also of this form. 

Our protocols proceed as follows: starting in a state $\ket{\psi(\vec{\omega};0)}$, after any given pass through the interferometers we use control unitaries to coherently switch between families of probe states such that the relative phase between the branches is preserved (that is, we change $\vec{\omega}$, but not $\varphi$). The fact that an optimal protocol must coherently map between such states is proven in Lemma~\ref{lemma:coherent-mapping} in Appendix~\ref{app:proof-of-lemma-1}.
We stay in the family of states $\ket{\psi(\vec\omega_n;\varphi)}$ for a fraction $p_n$ of the passes where $p_n= \frac{r_n}{M}$ for $r_n\in\{0, 1, \cdots, M\}$ such that $\sum_n p_n=1$. Here $n$ indexes some enumeration of the families of states in $\mathcal{T}$. 

The value of the component $\omega_j$ in a given probe state determines the contribution of the parameter $\theta_j$ coupled to sensor $j$ to the relative phase between the two branches of the probe state during a single pass. In particular, in a single pass with a probe state in the family $\ket{\psi(\vec\omega; \varphi)}$, the relative phase between the two branches of the probe state becomes $\vec\omega\cdot\vec\theta+\varphi$. Assuming an initial probe state with $\varphi=0$ and summing over all passes we obtain a total relative phase 
\begin{align}
   \varphi_\mathrm{tot}&= M\sum_n p_n(\vec\omega_n\cdot\vec\theta)\\
   &=: (W\vec r)\cdot\vec\theta.
\end{align}
In the second line, we implicitly defined $W$ as a matrix whose columns are the vectors $\vec\omega_n\in\mathcal{W}$ and $\vec r:=M\vec p\in\mathbb{Z}^{|\mathcal{T}|}$. Explicitly computing the Fisher information matrix for these states demonstrates that the optimality condition in Eq.~(\ref{eq:finalcond}) is satisfied if
\begin{equation}\label{eq:main_equation}
 W\vec r=NM\frac{\vec\alpha}{\norm{\vec\alpha}_{1,\mathcal{P}}};
\end{equation}
see \cref{app:protocols} for details.
Consequently, any integer solution $\vec r$ to Eq.~(\ref{eq:main_equation}) such that 
\begin{align}\label{eq:main_equation_constraint}
\norm{\vec r}_1&=M, \nonumber \\
\vec r&\geq 0,
\end{align}
yields an optimal protocol. The protocols of Ref.~\cite{proctor2017networked}, described above and generalized in \cref{app:protocols}, are a particularly simple case within this class with $M=1$ and $\vec\omega=\frac{N\vec\alpha}{\norm{\vec\alpha}_{1,\mathcal{P}}}$, i.e. we select out only a single column of $\mathcal{W}$.  
 
 Solutions to Eqs.~(\ref{eq:main_equation}) and(\ref{eq:main_equation_constraint}) are not guaranteed to exist for all $N,M$. In particular, we require that 
\begin{equation}\label{eq:diophantineNM}
    NM\frac{\vec\alpha}{\norm{\vec\alpha}_{1,\mathcal{P}}}\in\mathbb{Z}^d.
\end{equation}
For $\vec\alpha\in\mathbb{Q}$ and sufficiently large $N$ or $M$ this hold true. Setting up the system of equations in Eqs.~(\ref{eq:main_equation})-(\ref{eq:main_equation_constraint}) that must be solved to pick out explicit protocols requires identifying the set of vectors $\mathcal{W}$ defined in Eq.~(\ref{eq:omegadef}). 
While computationally straightforward, if expensive, to construct and enumerate this set, the number of states is extremely large, yielding a correspondingly large set of linear Diophantine equations in Eq.~(\ref{eq:main_equation}). Consequently, it is reasonable to place further, experimentally motivated constraints to limit this set of states and pick out advantageous protocols. For instance, one such constraint is to limit the amount of entanglement between modes on any given pass. We consider this case in the following section. 

It is also important to note that integer linear programming is NP-hard~\cite{arora2009computational}, so finding a particular solution once we add additional constraints is not a computationally easy task. Regardless, in applications one can apply standard (possibly heuristic) algorithms for integer linear programming to seek solutions. If a solution is found, it is known to be optimal. Consequently, proving the existence or lack thereof of a solution with certain additional constraints may be intractable for large problem instances. 

Similar arguments to those that go into proving Lemma~\ref{lemma:opt_states} allow us to show that, for quadrature sensing, the condition in Eq.~(\ref{eq:finalcondquad}) can be reduced to the condition that
\begin{equation}
\mathcal{F}(\vec\theta)_{ij}\sim\frac{4\overline{N}t^2}{\norm{\vec\alpha}_2^2}\alpha_i\alpha_j,
\end{equation}
which is proven in \cref{app:fim_quad}.
However, there is not a clearly interesting family of states that can be leveraged to achieve this quantum Fisher information, as in the case of number operator coupling or qubit sensors~\cite{ehrenberg2023minimum}. However, the existing optimal protocols described above do obey this condition asymptotically in average photon number $\overline{N}$.

\section{Entanglement requirements}\label{sec:entanglement}
The remaining flexibility in the choice of optimal probe states enabled by some control also allows us to impose further experimentally relevant constraints. One reasonable constraint is the amount of intermode entanglement required during the sensing process. This was considered in Ref.~\cite{ehrenberg2023minimum} for the case of qubit sensors. 

The answer to the entanglement question in the current context depends crucially on the sorts of control operations we allow. In the number operator case, with arbitrary time-dependent control, only two-mode entanglement is needed at any given time, as one can simply prepare a N00N state between the reference and one of the sensing modes and coherently switch which sensing mode is entangled with the reference mode such that the time spent entangled with mode $j$ is given by $t_j=|\alpha_j|t/\norm{\vec\alpha}_1$. For similar reasons, no entanglement is needed for displacement sensing; here, no reference mode is needed and one can simply sequentially apply displacement operators for a time $t_j=|\alpha_j|t/\norm{\vec\alpha}_1$ on a single-mode squeezed state, followed by a homodyne measurement. 
When control operations to change the probe state are allowed only at $M$ discrete time intervals, as described by Eq.~(\ref{eq:controlHnumber}), the problem becomes more interesting. For number operator coupling, subject to a fixed photon number constraint, any optimal protocol requires at least ($\lceil\norm{\vec\alpha}_0/M\rceil$+1)-mode entanglement. This bound is fairly trivial: it merely states that one must be entangled with each nontrivial mode for at least one pass. For displacement operator coupling, subject to a fixed average photon number constraint, an essentially identical argument allows us to prove that any optimal protocol requires at least $\lceil\norm{\vec\alpha}_0/M\rceil$-mode entanglement. The difference of one is because, unlike displacement sensing, phase sensing generally requires entanglement with a reference mode. In the $M\rightarrow \infty$ limit, we recover the continuous control case, so these trivial bounds can be tight. This triviality is in contrast to the qubit case, where results analogous to Lemma~\ref{lemma:opt_states} lead to significantly tighter constraints on the minimum amount of necessary entanglement for optimal protocols~\cite{ehrenberg2023minimum}. This discrepancy arises due to the fact that, unlike with photonic resources which must be distributed in a zero-sum way between modes, for qubit sensors one can be maximally sensitive to all coupled parameters simultaneously.

\begin{table*}
    \centering 
    \begin{tabular}{m{0.18\linewidth}|c|c|c}
                                 &  Qubit phase sensing & Phase sensing  & Displacement sensing \\
    \hline\hline
    Parameter coupling         & $\frac{1}{2}\hat{\sigma}_i^z\theta_i$ & $\hat{n}_i \theta_i$  & $\frac{i}{2}(\hat{a}^\dagger_i - \hat{a}_i)\theta_i$ \\ \hline
    Resources         & Qubit number, $d$ & Photon number, $N$  & Avg. photon number, $\overline{N}$\\ 
                     & sensing time, $t$ & sensing time, $t$  & sensing time, $t$\\ \hline
    MSE (separable)  & $\geq \frac{\norm{\vec\alpha}_2^2}{t^2}$~\cite{eldredge2018optimal}   &  $\geq \frac{\norm{\vec\alpha}_{2/3}^2}{N^2t^2}$~\cite{proctor2017networked}      &  $\geq\frac{\norm{\vec\alpha}_1^2}{4\overline{N}t^2}$~         \\\hline
    MSE (entangled)  & $\geq \frac{\norm{\vec\alpha}_\infty^2}{t^2}$~\cite{eldredge2018optimal} &  $\geq \frac{\norm{\vec\alpha}_{1,\mathcal{P}}^2}{N^2t^2}$~ & $\geq\frac{\norm{\vec\alpha}_2^2}{4\overline{N}t^2}$~\cite{xia2019repeater}~ \\\hline
    Entanglement needed (discrete controls) & $k \geq \max\left\{\left\lceil\frac{\norm{\vec\alpha}_1}{\norm{\vec\alpha}_\infty}\right\rceil, \left\lceil\frac{\norm{\vec\alpha}_0}{M}\right\rceil \right\}$   & $k>\left\lceil\frac{\norm{\vec\alpha}_0}{M}\right\rceil$~& $k\geq\left\lceil\frac{\norm{\vec\alpha}_0}{M}\right\rceil$~\\\hline
    Entanglement needed (arbitrary controls) & $\frac{\norm{\vec\alpha}_1}{\norm{\vec\alpha}_\infty}\in (k-1, k]$~\cite{ehrenberg2023minimum} & $k=2$~& No entanglement \\\hline
    $k$-partite {entanglement} {protocol} {always} {exists?} & Yes~\cite{ehrenberg2023minimum} & No ~ & Yes  \\
    \hline\hline
    \end{tabular}
    \caption{Comparison of the lower bounds on the mean square error and entanglement requirements for an (asymptotically) optimal protocol obeying the corresponding conditions on the quantum Fisher information for the task of estimating a linear function $q=\vec\alpha\cdot\vec\theta$ with qubit, phase sensing, and displacement sensing quantum sensor networks.}
    \label{tab:comp_qubit_photon}
\end{table*}

\section{Conclusion and outlook}\label{sec:conclusion}
 We have determined the fundamental achievable performance limits for phase sensing and have extended proofs of lower bounds for displacement sensing beyond just an average to arbitrary functions. In the process, we proved a long-standing conjecture regarding function estimation with number operator coupling~\cite{proctor2017networked} and showed that some of the protocols that exist in the literature~\cite{proctor2017networked, proctor2018multiparameter,zhuang2018distributed}, are, in fact, optimal in the asymptotic limit. By considering different implementations of a quantum sensor network within a single framework, we reveal the role of entanglement and controls as they relate to the type of coupling and whether the relevant resource is ``parallel'' (as in qubit sensor networks, where all parameters can simultaneously be measured to maximal precision) or ``sequential'' (as in photonic sensor networks, where the photons must be optimally distributed between modes). Our approach to proving our bounds also enables an algebraic framework for developing further optimal protocols, subject to various constraints. Here, we considered the particular case of entanglement-based constraints, enabling comparison to similar work in the case of qubit sensors~\cite{ehrenberg2023minimum}. 
These results, and how they fit into the landscape of known results for quantum sensor networks, are summarized in Table~\ref{tab:comp_qubit_photon}. How other constraints impact the existence of and control requirements for optimal protocols remains an interesting open question deserving of further study.

\begin{acknowledgements}
We thank Luis Pedro Garc\'ia-Pintos for helpful discussions. This work was supported in part by DARPA SAVaNT ADVENT, AFOSR MURI, AFOSR,  NSF STAQ program, DoE ASCR Accelerated Research in Quantum Computing program (award No.~DE-SC0020312),  NSF QLCI (award No.~OMA-2120757), and the DoE ASCR Quantum Testbed Pathfinder program (awards No.~DE-SC0019040 and No.~DE-SC0024220). Support is also acknowledged from the U.S.~Department of Energy, Office of Science, National Quantum Information Science Research Centers, Quantum Systems Accelerator.
\end{acknowledgements}

\onecolumngrid

\begin{appendix}

\section{Bound for Local Phase Shifts}\label{app:bnd-phase}
In this appendix, we derive lower bounds for the mean square error of measuring a linear function $q(\vec\theta)=\vec\alpha\cdot\vec\theta$ of local phase shifts, generated via coupling to the number operator $\hat{n}_j$, as specified by the Hamiltonian in Eq.~(\ref{eq:initialH}) and Eq.~(\ref{eq:particle_number_coupling}). 

In particular, we seek to solve the optimization problem in Eq.~(\ref{eq:minprob}), restated here for convenience:
\begin{align}\label{eq:minprob_app}
     \min_{\vec\beta} \max_\rho [\Delta(\vec\beta\cdot\hat{\vec{g}})_\rho]^2, \quad
    \text{subject to }\, \vec\alpha\cdot\vec\beta=1.
\end{align}
Here, $\hat{\vec{g}}=\hat{\vec{n}}=(\hat{n}_1, \hat{n}_2, \cdots, \hat{n}_d)^T$. For fixed particle number $N$, the Hilbert space on which possible probe states $\rho$ are defined is finite dimensional, and it holds that~\cite{boixo2007generalized}
\begin{equation}\label{eq:seminorm_bound}
[\Delta(\vec\beta\cdot\hat{\vec{n}})_\rho]^2\leq \frac{\norm{\vec\beta\cdot\hat{\vec{n}}}_{s,N}^2}{4},
\end{equation}
where $\norm{\vec\beta\cdot\hat{\vec{n}}}_{s, N}$ is the Fock-space-restricted seminorm of $\vec{\beta}\cdot\hat{\vec{n}}$ (defined as the difference between the maximum and minimum eigenvalues of $\vec\beta\cdot\hat{\vec{n}}$ restricted to the $N$-photon subspace). As we want to maximize the quantum Fisher information with respect to the choice of probe state $\rho$, and because \cref{eq:seminorm_bound} is saturable when $\rho$ is an equal superposition of the eigenstates of $\vec{\beta}\cdot\hat{\vec{n}}$ with maximum and minimum eigenvalues, we can consider the following optimization problem:
\begin{align}\label{eq:minprob_app2}
    \text{minimize (w.r.t. $\vec\beta$) }&\norm{\vec\beta \cdot \hat{\vec n}}_{s,N}, \nonumber\\
    \text{subject to }&\, \vec\alpha\cdot\vec\beta=1.
\end{align}
To begin, note that the largest eigenvalue of $\vec\beta\cdot\hat{\vec n}$ in the $N$-particle subspace is given by
\begin{equation}\label{eq:lambdamax}
    \lambda_\mathrm{max}(\vec\beta\cdot\hat{\vec n})=N\max\big\{\max_j \beta_j, 0\big\}=:N\beta_\mathrm{max},
\end{equation}
where we have implicitly defined $\beta_\mathrm{max}$. This largest eigenvalue corresponds to the eigenstate that consists of placing all photons in the mode corresponding to the largest positive $\beta_j$. If all $\beta_j\leq 0$, the largest eigenvalue is zero, obtained by any state with no particles in the sensor modes. Note that this requires the use of an extra mode (an ancilla or so-called ``reference mode'') to ``store'' these photons, as we fix the total photon number of our state to be $N$. 

Similarly, the smallest eigenvalue of $\vec\beta\cdot\hat{\vec n}$ in the $N$-particle subspace is given by
\begin{equation}\label{eq:lambdamin}
    \lambda_\mathrm{min}(\vec\beta\cdot\hat{\vec n})=N\min\big\{\min_j \beta_j, 0\big\}=:N\beta_\mathrm{min},
\end{equation}
where we have implicitly defined $\beta_\mathrm{min}$.

Using the facts above about the maximum and minimum eigenvalues of $\vec\beta\cdot\hat{\vec n}$ in the $N$-particle subspace we can rewrite the optimization problem in Eq.~(\ref{eq:minprob}) as
\begin{align}\label{eq:minprob2}
    \text{minimize }&N\left(\beta_\mathrm{max} - \beta_\mathrm{min} \right),\nonumber\\
    \text{subject to }&\, \vec\alpha\cdot\vec\beta=1.
\end{align}
As in the main text, define $\mathcal{P}:=\{j \,|\, \alpha_j\geq 0\}$ and $\mathcal{N}:=\{j \,|\, \alpha_j< 0\}$. We then have the following lemma.
\begin{lemma}\label{lemma:beta_signs}
The solution $\vec\beta^*$ to Eq.~(\ref{eq:minprob2}) is such that $\beta^*_j\geq 0$ for all $j\in \mathcal{P}$,  and $\beta^*_j\leq 0$ for all $j\in \mathcal{N}$. That is, $\alpha_{j}\beta_{j}^{*} \geq 0$ for all $j$. 
\end{lemma}
\begin{proof}
We proceed by contradiction. Let $\mathcal{J}_-=\{j \, |\, \alpha_j\beta^*_j<0\}$ and $\mathcal{J}_+=\{j \, |\, \alpha_j\beta^*_j\geq 0\} $.  Suppose the solution vector $\vec\beta^*$ to Eq.~(\ref{eq:minprob2}) has $\mathcal{J}_-\neq\emptyset$. We can construct an alternative candidate solution vector $\vec\beta'$  as follows: First, let $\vec\beta'=\vec\beta^*$. Then set $\beta_j'=0$ for all $j\in\mathcal{J}_-$. 
In order to still satisfy the constraint $\vec{\alpha}\cdot\vec{\beta'}=1$, we must reduce the values of some other components in $\vec\beta'$. In particular, it is simple to calculate that a valid solution is, for $j \in \mathcal{J}_{+}$,
\begin{equation}
    \beta_{j}' =  \frac{\beta_{j}^{*}}{\sum_{j\in\mathcal{J}_{+}}\alpha_{j}\beta_{j}^{*}} = \frac{\beta_{j}^{*}}{1-\sum_{j\in\mathcal{J}_{-}}\alpha_{j}\beta_{j}^{*}}.
\end{equation}
Again, when $j \in \mathcal{J}_{-}$, $\beta'_{j} = 0$.

Let $\beta_\mathrm{max}':=\max\big\{\max_j \beta'_j, 0\big\}$ and $\beta_\mathrm{min}':=\max\big\{\min_j \beta'_j, 0\big\}$. By construction, $\beta_\mathrm{max}'\leq\beta_\mathrm{max}^{*}$ and $0=\beta_\mathrm{min}'\geq\beta_\mathrm{min}^{*}$. Consequently, $\vec\beta'$ yields a smaller solution candidate than $\vec\beta^*$. This contradicts the fact that $\vec\beta^*$ is the optimal solution. The lemma statement follows as an immediate consequence.
\end{proof}

Lemma~\ref{lemma:beta_signs} allows us to rewrite the minimization problem in Eq.~(\ref{eq:minprob2}) once again as
\begin{align}\label{eq:minprob3}
    \text{minimize }&N\left[\max_{j \in \mathcal{P}} \beta_j - \min_{j \in \mathcal{N}} \beta_j \right],\nonumber\\
    \text{where }&\, \beta_j \geq 0 \,\,\forall\, j \in \mathcal{P}, \nonumber\\
    &\beta_j \leq 0 \,\,\forall\, j \in \mathcal{N},\nonumber \\
    \text{subject to }&\, \vec\alpha\cdot\vec\beta=1.
\end{align}
In the above, we define $\max_{j \in \mathcal{P}}\beta_j$ ($\min_{j \in \mathcal{N}} \beta_j$) to be zero if $\mathcal{P}=\emptyset$ ($\mathcal{N}=\emptyset$). A further simplification is enabled by another lemma.
\begin{lemma}\label{lemma:beta_equality}
The solution vector $\vec\beta^*$ to Eq.~(\ref{eq:minprob3}) is such that $\beta^*_j=\beta^*_\mathrm{max}$ for all $j\in\mathcal{P}$ and $\beta^*_j=\beta^*_\mathrm{min}$ for all $j\in\mathcal{N}$.
\end{lemma}
\begin{proof}
We proceed by contradiction. Suppose the solution vector $\vec\beta^*$ is such that $\beta_i^*\neq\beta^*_j$ for some $i,j \in \mathcal{P}$. Then we could consider an alternative candidate solution vector $\vec\beta'$ where $\beta_k'= \frac{\sum_{l \in \mathcal{P}} \alpha_l \beta_l^*}{\sum_{l \in \mathcal{P}} \alpha_l}$ for all $k\in \mathcal{P}$. Similarly, if $\beta_i^*\neq\beta^*_j$ for some $i,j \in \mathcal{N}$ we could consider $\beta_k'= \frac{\sum_{l \in \mathcal{N}} \alpha_l \beta_l^*}{\sum_{l \in \mathcal{N}} \alpha_l}$ for all $k\in \mathcal{N}$. Clearly, $\vec\beta'$ still satisfies the constraint
\begin{align}
    \vec\alpha\cdot\vec\beta' = \sum_{m\in \mathcal{P}}\alpha_m  \left(\frac{\sum_{l \in \mathcal{P}} \alpha_l \beta_l^*}{\sum_{l \in \mathcal{P}} \alpha_l}\right)+\sum_{m\in \mathcal{N}}\alpha_m  \left(\frac{\sum_{l \in \mathcal{N}} \alpha_l \beta_l^*}{\sum_{l \in \mathcal{N}} \alpha_l}\right)=\vec\alpha\cdot\vec\beta^*=1.
\end{align}
Additionally, $\beta'$ also clearly still has $\beta'_{j}\geq0$ when $j \in \mathcal{P}$ and $\beta'_{j}\leq 0$ when $j \in \mathcal{N}$. 
But, by construction (because the weighted average of a set is less than its maximum element), 
\begin{equation}
    N\left[\max_{j \in \mathcal{P}} \beta_j' - \min_{j \in \mathcal{N}} \beta_j' \right]<N\left[\max_{j \in \mathcal{P}} \beta_j^* - \min_{j \in \mathcal{N}} \beta_j^* \right].
\end{equation}
So $\vec\beta^*$ is not the solution vector and we have arrived at a contradiction.
\end{proof}

As a direct consequence of Lemma~\ref{lemma:beta_equality} we can rewrite the optimization problem in  Eq.~(\ref{eq:minprob3}) one last time as
\begin{align}
    \text{minimize (w.r.t. $\beta_\mathrm{min}, \beta_\mathrm{max}$) }&N \left[\beta_\mathrm{max} - \beta_\mathrm{min}\right],\\
    \text{subject to }&\,\beta_\mathrm{max} \geq 0, \beta_\mathrm{min} \leq 0,\\
    &\, \beta_\mathrm{max} \sum_{j \in \mathcal{P}} \alpha_j + \beta_\mathrm{min} \sum_{j \in \mathcal{N}} \alpha_j = 1.
\end{align}

Because this is a linear objective function, the optimal solution will be one of the two boundary solutions: $\beta_\mathrm{max} = \frac{1}{\sum_{i \in \mathcal{P}} \alpha_i}, \beta_\mathrm{min} = 0 $  or $\beta_\mathrm{min} = \frac{1}{\sum_{i \in \mathcal{N}} \alpha_i}, \beta_\mathrm{max} = 0 $. Minimizing over these two candidate solutions, we obtain the final result
\begin{equation}
 \norm{\hat g_q}_{s, N}^2=\frac{N^2}{\max(\sum_{i \in \mathcal{P}} \alpha_i, \sum_{i \in \mathcal{N}} \alpha_i)^2}.
\end{equation}
Consequently, via the quantum Cram\'{e}r-Rao bound, \cref{eq:bnd_betastar},
\begin{align}
\mathcal{M}&\geq \frac{\max\left\{\sum_{i \in \mathcal{P}} \alpha_i, \sum_{i \in \mathcal{N}} \alpha_i\right\}^2}{N^2t^2} \nonumber \\
    &=:  \frac{\max\left\{\norm{\vec\alpha}_{1,\mathcal{P}}^2, \norm{\vec\alpha}_{1,\mathcal{N}}^2\right\}}{N^2t^2},
\end{align}
which is Eq.~(\ref{eq:finalbnd}), and where $||\vec{\alpha}||_{1,\mathcal{P}}$ and $||\vec{\alpha}||_{1,\mathcal{N}}$ are the one-norm restricted to positive and negative values, respectively, of $\vec{\alpha}$.
In the special case of all positive coefficients (i.e., $\mathcal{N}=\emptyset$), this reduces to
\begin{equation}
    \mathcal{M}\geq \frac{\norm{\vec\alpha}_1^2}{N^2t^2},
\end{equation}
which, as described in the main text, proves a conjecture from Ref.~\cite{proctor2017networked} that this is the minimum attainable variance for $\vec\alpha\in\mathbb{Q}^d$ with $\vec\alpha\geq 0$.

\section{Bound for Local Displacements}\label{app:bnd-displacement}
In this appendix, we derive Eq.~(\ref{eq:dispbnd}) for the mean square error attainable for measuring a linear function of local displacements, restricting to probe states with fixed average photon number $\overline{N}$.

\subsection{Separable Bound}
To begin, it is helpful to present the bound for the more restricted case where we use separable input states. Begin by considering the lower bound on the variance of measuring a displacement $\varphi$ coupled to a single mode via $H=\varphi \hat{p}$, following the proof sketched in Ref.~\cite{zhang2021distributed}. The quantum Fisher information is given by
\begin{equation}
    \mathcal{F}(\varphi)=4[\Delta(\hat p)_\rho]^2,
\end{equation}
where $\rho$ is the probe state, which is restricted to have an average photon number $\overline{N}$. An initial displacement does not enhance precision~\cite{zhang2021distributed}, so we can consider zero-mean displacement input states. For such probe states,
\begin{align}
(\Delta\hat p)^2&=-\frac{1}{4}\langle (\hat a^\dagger-\hat a)^2\rangle=-\frac{1}{4}(\langle \hat a^\dagger \hat a^\dagger\rangle -\langle \hat a^\dagger \hat a\rangle -\langle \hat a \hat a^\dagger \rangle +\langle \hat a \hat a\rangle),\\
(\Delta\hat x)^2&=\frac{1}{4}\langle (\hat a^\dagger+\hat a)^2\rangle=\frac{1}{4}(\langle \hat a^\dagger \hat a^\dagger\rangle +\langle \hat a^\dagger \hat a\rangle +\langle \hat a \hat a^\dagger \rangle +\langle \hat a \hat a\rangle),
\end{align}
so that
\begin{align}\label{eq:numberconstraint}
    \overline{N}=\langle \hat a^\dagger \hat a\rangle = (\Delta\hat p)^2+(\Delta\hat x)^2-\frac{1}{2},
\end{align}
where we used that $\hat a \hat a^\dagger=\hat a^\dagger \hat a+1$.  We can then use the uncertainty principle
\begin{equation}
    (\Delta\hat p)^2(\Delta\hat x)^2\geq \frac{1}{16},
\end{equation}
which follows from our definition of the quadrature operators as $\hat x=(\hat a^\dagger+\hat a)/2$ and $\hat p = i(\hat a^\dagger -\hat a)/2$. Therefore, 
\begin{align}
    \xi\left(\overline{N}-\xi+\frac{1}{2}\right)\geq \frac{1}{16},
\end{align}
where we let $\xi:=(\Delta\hat p)^2$. Then
\begin{align}
    -16\xi^2 +(16\overline{N}+8)\xi-1\geq 0.
\end{align}
To maximize $\xi$, this inequality must be saturated, so we can solve the corresponding quadratic to obtain the solution
\begin{equation}
\xi=\frac{-8(2\overline{N}+1)+\sqrt{64(2\overline{N}+1)^2-64}}{-32}\implies 4\xi=(\sqrt{\overline{N}}+\sqrt{\overline{N}+1})^2\sim 4\overline{N}.
\end{equation}
It is worth noting that the $\mathcal{O}(\overline{N})$ asymptotic behavior of the maximum variance of $\hat p$ could have been obtained with no calculation just from examining the constraint in Eq.~(\ref{eq:numberconstraint}) under the assumption that $(\Delta\hat{x})^2$ can be made negligibly small.

Putting everything back together, we have found that, optimizing over states with fixed average photon number $\overline{N}$, the following holds:
\begin{equation}\label{eq:single_mode_displacement_avg_photon}
[\Delta(\tilde\varphi)]^2\geq \frac{1}{\mathcal{F}}\geq \frac{1}{t^2(\sqrt{\overline{N}}+\sqrt{\overline{N}+1})^2}=\frac{1}{4t^2\overline{N}}+\mathcal{O}\bigg(\frac{1}{t^2\overline{N}^{2}}\bigg).
\end{equation}

Working in the asymptotic in $\overline{N}$ limit, we can use Eq.~(\ref{eq:single_mode_displacement_avg_photon}) to obtain a bound on performance for estimating a linear function $q(\vec\theta)=\vec\alpha\cdot\vec\theta$ with an unentangled protocol as 
\begin{equation}\label{eq:opt_avg_photon_unent}
(\Delta\tilde q)^2\geq \frac{1}{t^2}\min_{\{\overline{N}_j\}} \sum_{j=1}^d\frac{|\alpha_j|^2}{4\overline{N}_j}+\mathcal{O}\left(\frac{1}{\overline{N}_j^{2}}\right),
\end{equation}
where $\overline{N}_j=\langle \hat a^\dagger_j \hat a_j\rangle$ is the average number of photons used in mode $j$ and $\sum_j \overline{N}_j=\overline{N}$. Assume without loss of generality that $|\alpha_j|>0$ for all $j$ (that is, no $\alpha_{j} = 0$) and independent of $\overline{N}$. Then we can optimize (at leading order in $\frac{1}{\overline{N}}$) the distribution of photons amongst the modes using the Lagrangian
\begin{equation}
\mathcal{L}=\sum_{j=1}^d \frac{|\alpha_j|^2}{4\overline{N}_j} + \gamma\left(\sum_{j=1}^d \overline{N}_j-\overline{N}\right),
\end{equation}
where $\gamma$ is a Lagrange multiplier. A bit of algebra yields that
\begin{align}
\frac{\partial\mathcal{L}}{\partial \overline{N}_j}=0\implies \overline{N}_j = \frac{|\alpha_j|}{2\sqrt{\gamma}}.
\end{align}
This further implies that
\begin{equation}
\overline{N}=\sum_{j=1}^d \overline{N}_j = \frac{\norm{\vec\alpha}_1}{2\sqrt{\gamma}},
\end{equation}
allowing us to obtain the optimal division of photons as
\begin{equation}
\overline{N}_j=\frac{|\alpha_j|}{\norm{\vec\alpha}_1}\overline{N}.
\end{equation}
We note that this solution is clearly the desired minimum of the Lagrangian, as maximizing the objective would lead to setting any $\overline{N}_{j}$ to 0. 
Plugging this back into Eq.~(\ref{eq:opt_avg_photon_unent}) we obtain the (asymptotic in $\overline{N}$) separable bound
\begin{equation}
[\Delta\tilde q]^2\geq \frac{\norm{\vec\alpha}_1^2}{4\overline{N}t^2}+\mathcal{O}\left(\frac{1}{\overline{N}^2}\right).
\end{equation}

This bound can be achieved by using the single-mode protocols in Ref.~\cite{zhang2021distributed} for each mode and then computing the function of interest classically as a linear combination of the individual estimators.

\subsection{General Function Estimation Bound}
In this subsection, we turn to our primary task: deriving Eq.~(\ref{eq:dispbnd}) for the mean square error attainable for measuring a linear function of local displacements, restricting to probe states with fixed average photon number $\overline{N}$. 

To derive this bound, we must solve the optimization problem in Eq.~(\ref{eq:minprob}) for $\hat{g}_j=\hat{p}_j$:
\begin{align}\label{eq:minprob_app2}
    \min_{\vec\beta} \max_{\rho} [\Delta(\vec\beta\cdot\hat{\vec{p}})_\rho]^2, \quad
    \text{subject to }\, \vec\alpha\cdot\vec\beta=1.
\end{align}
We can write
\begin{align}
[\Delta(\vec\beta\cdot\hat{\vec{p}})]^2&=\sum_{i,j =1}^{d}\beta_i\beta_j\mathrm{Cov}(\hat p_i, \hat p_j)\nonumber \\
&\leq \sum_{i,j=1}^{d}\beta_i\beta_j\sqrt{(\Delta\hat{p}_i)^2(\Delta\hat{p}_j})^2\nonumber \\
&= \left[\sum_{j=1}^{d}\beta_j\Delta\hat{p}_j\right]^2\nonumber \\
&\leq \norm{\vec\beta}_2^2\sum_{j=1}^{d} (\Delta\hat{p}_j)^2,
\end{align}
where we applied the Cauchy-Schwarz inequality twice. Using the same assumption of zero-displacement states we made in the previous section, we can further bound $\sum_j (\Delta\hat{p}_j)^2$ using the constraint on average photon number 
\begin{equation}
\sum_{j=1}^{d} \left[(\Delta\hat{p}_j)^2+(\Delta\hat{x}_j)^2\right] - \frac{d}{2}=\sum_{j=1}^{d} \langle a^\dagger_j a_j\rangle =\overline{N},
\end{equation}
implying that
\begin{equation}\label{eq:a}
\sum_{j=1}^{d} (\Delta\hat{p}_j)^2 \leq \overline{N}+\frac{d}{2}.
\end{equation}
Equation~(\ref{eq:a}) is tight when $(\Delta \hat x_j)^2=0$ for all $j$. This is, of course, impossible to achieve, but can be approached asymptotically with increasing $\overline{N}$ ($\overline{N}\gg d$). Furthermore, using the fact that $\vec\alpha$ is dual to $\vec\beta$ and the Cauchy-Schwarz inequality, it holds that
\begin{equation}
1=\vec\alpha\cdot\vec\beta\leq \norm{\vec\beta}_2 \norm{\vec\alpha}_2.
\end{equation}
As we want to minimize with respect to $\vec\beta$, we consider the case where this inequality is saturated (i.e. $\vec\beta^*=\frac{\vec\alpha}{\norm{\vec\alpha}_{2}^{2}}$). Therefore, $\norm{\vec\beta^*}_2=\frac{1}{\norm{\vec\alpha}_2}$, and we obtain
\begin{equation}
[\Delta(\vec\beta\cdot\hat{\vec{p}})]^2\leq \frac{\overline{N}}{\norm{\vec\alpha}_2^2}+\mathcal{O}\left(\frac{d}{\norm{\vec\alpha}_2^2}\right).
\end{equation}
This yields the final bound
\begin{equation}
\mathcal{M}\geq \frac{\norm{\vec\alpha}_2^2}{4\overline{N}t^2}-\mathcal{O}\left(\frac{d\norm{\vec\alpha}_2^2}{\overline{N}^2t^2}\right).
\end{equation}
From the derivation alone, it is not obvious that this bound can be saturated, but the existence of protocols that achieve it~\cite{zhuang2018distributed} indicate that this bound is, indeed, tight asymptotically in $\overline{N}$.

\section{Quantum Fisher Information Matrix Elements}
In this appendix, we derive the matrix elements of the quantum Fisher information matrix for generators $\hat{n}_{j}$ and $\hat{p}_{j}$ under the unitary evolution \cref{eq:full_unitary}. 
For number operator coupling $\hat{g}_j=\hat{n}_j$,
\begin{align}\label{eq:generator}
\mathcal{H}_j=-iU^\dagger\partial_j U&=-\sum_{m=1}^M\left(\prod_{l=1}^{m-1}U^{(l)} V\right)^\dagger \hat{n}_j \left(\prod_{l=1}^{m-1}U^{(l)} V\right)\nonumber \\
&=:-\sum_{m=1}^M \hat{n}_j(m),
\end{align}
where in the second line we implicitly defined $\hat{n}_j(m)$.
Consequently, we can compute the quantum Fisher information matrix elements via Eq.~(\ref{eq:fim}) to be
\begin{align}\label{eq:fim_general}
\mathcal{F}(\vec\theta)_{ij}=4\Bigg[&\sum_{l=1}^M\sum_{m=1}^M\frac{1}{2}\langle\{\hat{n}_i(l),\hat{n}_j(m)\}\rangle-\left(\sum_{m=1}^M\langle\hat{n}_i(m)\rangle\right)\left(\sum_{m=1}^M\langle\hat{n}_j(m)\rangle\right)\Bigg].
\end{align}

When $\hat{U}^{(j)}=I$ for all $j$, this reduces to
\begin{align}\label{eq:fim_nohc}
    \mathcal{F}(\vec\theta)_{ij}&=4M^2\left[\langle\hat{n}_i\hat{n}_j\rangle-\langle\hat{n}_i\rangle\langle\hat{n}_j\rangle\right].
\end{align}

For quadrature operator coupling $\hat{g}_j=\hat{p}_j$, essentially identical manipulations yield
\begin{align}\label{eq:fim_general_quad}
\mathcal{F}(\vec\theta)_{ij}=4\Bigg[&\sum_{l=1}^M\sum_{m=1}^M\frac{1}{2}\langle\{\hat{p}_i(l),\hat{p}_j(m)\}\rangle-\left(\sum_{m=1}^M\langle\hat{p}_i(m)\rangle\right)\left(\sum_{m=1}^M\langle\hat{p}_j(m)\rangle\right)\Bigg],
\end{align}
where $\hat{p}_j(l)$ is defined as in Eq.~(\ref{eq:generator}) with $\hat{n}_j\rightarrow\hat{p}_j$.

\section{Protocols for Local Phase Shifts}\label{app:protocols}
In this appendix, we elaborate on the families of optimal protocols for measuring a linear function of phase shifts that we described in Sec.~IV. 

\subsection{An Optimal Protocol for Functions with Positive Coefficients}
We begin by reviewing a protocol from Ref.~\cite{proctor2017networked} for the special case of a linear function with positive coefficients (i.e., $\vec\alpha\geq 0$). Our results in \cref{app:bnd-phase} show that, as those authors conjectured, this protocol is optimal. In particular, consider using as the probe state a so-called proportionally weighted N00N state over $d+1$ modes:
\begin{align}
    \ket{\psi}\propto \bigg|N\frac{\alpha_1}{\norm{\vec\alpha}_1}, \cdots, N\frac{\alpha_d}{\norm{\vec\alpha}_1}, 0\bigg \rangle+\bigg|0,\cdots, 0,N\bigg\rangle,
\end{align}
where we have expressed the state in an occupation number basis over $d+1$ modes and have dropped the normalization for concision. The last mode serves as a reference mode. Observe that, for this state to be well defined, it is essential that $\frac{\vec\alpha}{\norm{\vec\alpha}_1}\in\mathbb{Q}^d$ and that $N$ is such that the resulting occupation numbers are integers, which may require that $N$ be large. 

Following imprinting of the parameters $\vec\theta$ onto the probe state via $M$ passes through the interferometers, one obtains
\begin{align}
    &\ket{\psi_M}=e^{-iM\hat{\vec n}\cdot\vec\theta}\ket{\psi}\propto \bigg|N\frac{\alpha_1}{\norm{\vec\alpha}_1}, \cdots, N\frac{\alpha_d}{\norm{\vec\alpha}_1}, 0\bigg \rangle+e^{i\vec\alpha\cdot\vec\theta \frac{NM}{\norm{\vec\alpha}_1}}\bigg|0,\cdots, 0,N\bigg \rangle.
\end{align}
This process allows us to saturate the bound in Eq.~(\ref{eq:finalbnd_positivealpha}). In particular, using Eq.~(\ref{eq:fim_general}) [which reduces to \cref{eq:fim_nohc} because there is no control required], it is straightforward to calculate that the quantum Fisher information matrix for the probe state is
\begin{equation}
    \mathcal{F}(\vec\theta)=\frac{(MN)^2}{\norm{\vec\alpha}_1^2}\vec\alpha\vec\alpha^T,
\end{equation}
which clearly satisfies the condition in Eq.~(\ref{eq:finalcond}) (recalling that $||\vec{\alpha}||_{1} = ||\vec{\alpha}||_{1, \mathcal{P}}$ here because we have assumed all coefficients are non-negative, and also recalling that $\Delta t = 1$ such that $M = t$). 

While the conditions on the quantum Fisher information matrix for an optimal protocol are met, a full protocol requires a description of the measurements used to extract the quantity of interest from the relative phase between the branches of $\ket{\psi_M}$. As described in the main text, this can be done via the robust phase estimation protocols of Refs.~\cite{kimmel2015robust,kimmel2015robusterratum,belliardo2020achieving} with a small multiplicative constant overhead relative to the quantum Cram\'{e}r-Rao bound (we also briefly discuss the idea behind robust phase estimation in Appendix~\ref{app:robust-phase-estimation}). The details of implementing the necessary parity measurements for N00N-like states are discussed in detail in Appendix A of Ref.~\cite{belliardo2020achieving} and Ref.~\cite{chiruvelli2011parity}.

\subsection{Extending the Optimal Protocol to Negative Coefficients}\label{s:optimalprotocols_neg}
While not explicitly considered in Ref.~\cite{proctor2017networked}, it is straightforward to extend the above protocol to the situation where $\mathcal{N}\neq \emptyset$, which we do here. Without loss of generality, assume the coefficients are ordered so that $\alpha_1\geq \alpha_2\geq \cdots \geq \alpha_d$. Using our standard assumption that $\norm{\vec\alpha}_{1,\mathcal{P}}\geq \norm{\vec\alpha}_{1,\mathcal{N}}$, we claim that the following probe state is optimal:
\begin{align}\label{eq:probe_state_neg_alpha}
    \ket{\psi}&\propto\bigotimes_{j\in\mathcal{P}}\ket{N\frac{\alpha_j}{\norm{\vec\alpha}_{1,\mathcal{P}}}}\ket{0}^{\otimes |\mathcal{N}|}\ket{0}+\ket{0}^{\otimes|\mathcal{P}|}\bigotimes_{j\in\mathcal{N}}\ket{N\frac{|\alpha_j|}{\norm{\vec\alpha}_{1,\mathcal{P}}}}\ket{N-N\frac{\norm{\vec\alpha}_{1,\mathcal{N}}}{\norm{\vec\alpha}_{1,\mathcal{P}}}},
\end{align}
where, again, the last mode is a reference mode, and we have dropped the normalization of the state. 
Interestingly, observe that, if $\norm{\vec\alpha}_{1,\mathcal{P}}= \norm{\vec\alpha}_{1,\mathcal{N}}$, the reference mode factors out and is unnecessary. Similar to the $\vec\alpha\geq 0$ case, for this state to be well defined, we require that $N|\alpha_j|/\norm{\vec\alpha}_{1,\mathcal{P}}\in\mathbb{N}$ for all $j$, which is always true for some sufficiently large $N$ provided $\vec\alpha\in\mathbb{Q}^d$.

Consider applying the encoding unitary for $M$ passes through the interferometers. For $\norm{\vec\alpha}_{1,\mathcal{P}}\geq \norm{\vec\alpha}_{1,\mathcal{N}}$, this yields
\begin{align}\label{eq:final_state_negative_alpha}
    &\ket{\psi_M}\propto\bigotimes_{j\in\mathcal{P}}\ket{N\frac{\alpha_j}{\norm{\vec\alpha}_{1,\mathcal{P}}}}\ket{0}^{\otimes |\mathcal{N}|}\ket{0}+e^{i\vec\alpha\cdot\vec\theta \frac{NM}{\norm{\vec\alpha}_{1,\mathcal{P}}}}\ket{0}^{\otimes|\mathcal{P}|}\bigotimes_{j\in\mathcal{N}}\ket{N\frac{|\alpha_j|}{\norm{\vec\alpha}_{1,\mathcal{P}}}}\ket{N-N\frac{\norm{\vec\alpha}_{1,\mathcal{N}}}{\norm{\vec\alpha}_{1,\mathcal{P}}}}.
\end{align}
This probe state is optimal in the sense of satisfying the Fisher information condition in Eq.~(\ref{eq:finalcond}). In the main text, we described an even more general family of protocols. Within this more general framework, we will prove this optimality. 

\subsection{A Family of Optimal Protocols}\label{s:protocol_family}
Finally, we describe a family of optimal protocols that satisfy the conditions on the quantum Fisher information matrix given in \cref{eq:finalcond}.
In the main text, we defined a family of optimal protocols in terms of vectors from the set
\begin{align}
    \mathcal{W}:=\left\{\vec\omega\in\mathbb{Z}^d\, \big|\, \norm{\vec\omega}_{1,\mathcal{P}}=N,\, \norm{\vec\omega}_{1,\mathcal{N}}\leq N,\, \omega_j\alpha_j\geq 0 \, \forall\, j\right\}.
\end{align}
In particular, from these vectors, we defined a set $\mathcal{T}$ of one-parameter families of probe states in an occupation number basis where each $\ket{\psi(\vec\omega;\varphi)}\in \mathcal{T}$ is labeled by a particular choice of $\vec\omega$ such that
\begin{align}\label{eq:omegastates}
    &\ket{\psi(\vec\omega;\varphi)}\propto \ket{\vec \omega|_\mathcal{P}}\ket{0}+e^{i\varphi}\ket{-\vec\omega|_{\mathcal{N}}}\ket{N-\norm{\vec\omega|_\mathcal{N}}_1},
\end{align}
where $\varphi\in\mathbb{R}$ is an arbitrary parameter and the last mode is a reference mode. Recall also that $\vec{\omega}_{\mathcal{P}}$ and $\vec{\omega}_{\mathcal{N}}$ are defined in \cref{eq:omega_restriction} as the restriction of $\vec{\omega}$ to $j \in \mathcal{P}$ and $\mathcal{N}$, respectively (for $j$ not in the correct set, the value is set to $0$). Note that such states are of the form of those in Lemma~\ref{lemma:opt_states}. We claimed that, by explicitly computing the Fisher information matrix for these states, one could demonstrate that the optimality condition in Eq.~(\ref{eq:finalcond}) is satisfied for a protocol such that
\begin{equation}\label{eq:main_equation-2}
 W\vec r=NM\frac{\vec\alpha}{\norm{\vec\alpha}_{1,\mathcal{P}}},
\end{equation}
where $\vec r\in\mathbb{Z}^{|\mathcal{T}|}$ is as defined in the main text and must obey the conditions
\begin{align}\label{eq:main_equation_constraint-2}
\norm{\vec r}_1&=M, \nonumber \\
\vec r&\geq 0.
\end{align}
Recall that $W$ is a matrix whose columns are the vectors $\vec{\omega}_n\in\mathcal{W}$.

Here we explicitly demonstrate this. We can easily evaluate
\begin{align}
\langle\hat{n}_j(m)\rangle=\bra{\psi(\vec\omega^{(m)};\varphi)}\hat n_j \ket{\psi(\vec\omega^{(m)};\varphi)}=\frac{|\omega^{(m)}_j|}{2}
\end{align}
and
\begin{align}
\langle\hat{n}_i(l)\hat{n}_j(m)\rangle&=\bra{\psi(\vec\omega^{(l)};\varphi)}\hat{n}_iU(m\leftrightarrow l)\hat n_j \ket{\psi(\vec\omega^{(m)};\varphi)}\nonumber \\
&=\frac{|\omega^{(l)}_i\omega^{(m)}_j|}{2}\bra{\psi_l(\vec\omega^{(l)};\varphi)}U(m\leftrightarrow l)\ket{\psi_m(\vec\omega^{(m)};\varphi)},
\end{align}
where $\hat{n}_j(m)$ are defined as in Eq.~(\ref{eq:generator}), and
\begin{equation}
U(m\leftrightarrow l)=\begin{cases}
\prod_{k=m}^{l-1} U^{(k)}V, & \text{if } l\geq m\\
\prod_{k=l}^{m-1} (U^{(k)}V)^\dagger, & \text{otherwise},
\end{cases}
\end{equation}
i.e., it is the unitary that converts between the $m$-th and $l$-th probe states. Additionally, $\vec{\omega}^{(m)}$ refers to the vector associated to the $m$-th probe state; correspondingly $\ket{\psi_l(\vec\omega^{(l)};\varphi)}$ is the branch of $\ket{\psi(\vec\omega^{(l)};\varphi)}$ with non-zero occupation number on mode $l$ and $\ket{\psi_m(\vec\omega^{(m)};\varphi)}$ is the branch of $\ket{\psi(\vec\omega^{(m)};\varphi)}$ with non-zero occupation number on mode $m$. For an optimal protocol, $U(m\leftrightarrow l)$ coherently maps the first (second) branch of $\ket{\psi(\vec\omega^{(l)};\varphi)}$ to the first (second) branch of $\ket{\psi(\vec\omega^{(m)};\varphi)}$; therefore, we have that the matrix element $\bra{\psi_l(\vec\omega^{(l)};\varphi)}U(m\leftrightarrow l)\ket{\psi_m(\vec\omega^{(m)};\varphi)}$ is nonzero if and only if the branches with non-zero occupation on modes $l$ and $m$ are the same. So we have that
\begin{align}
\langle\hat{n}_i(l)\hat{n}_j(m)\rangle=\frac{|\omega^{(l)}_i\omega^{(m)}_j|}{2} \xi_{ij},
\end{align}
where
\begin{equation}
\xi_{ij}:=\begin{cases}
1, & \text{if } i,j\in\mathcal{P} \text{ or } i,j\in\mathcal{N} \\
0, & \text{otherwise.}
 \end{cases}
\end{equation}

Putting everything together we obtain that
\begin{equation}\label{eq:fim_protocol_family}
\mathcal{F}(\vec\theta)_{ij}=(-1)^{\xi_{ij}+1}\left(\sum_{m=1}^M |\omega^{(m)}_i| \right)\left(\sum_{m=1}^M |\omega^{(m)}_j| \right).
\end{equation}

To prove the protocols work, we need to show that this Fisher information matrix obeys the condition in Eq.~(\ref{eq:finalcond}).
Without loss of generality, consider the case that $\norm{\vec\alpha}_{1,\mathcal{P}}\geq \norm{\vec\alpha}_{1,\mathcal{N}}$. 
We have that
\begin{align}
\sum_{j \in \mathcal{P}} \mathcal{F}(\vec\theta)_{ij} &= \mathrm{sgn}(\alpha_i) \left(\sum_{m=1}^M |\omega^{(m)}_i| \right)MN,
\end{align}
where we used that $\norm{\vec\omega}_{1,\mathcal{P}}=N$. So, to obey the condition in Eq.~(\ref{eq:finalcond}), we require that
\begin{equation}
\sum_{m=1}^M |\omega^{(m)}_i| =MN\frac{|\alpha_i|}{\norm{\vec\alpha}_{1,\mathcal{P}}}.
\end{equation}
Or, in vector form:
\begin{equation}\label{eq:fim_cond_protocol_family_app}
\sum_{m=1}^M |\vec\omega^{(m)}| =MN\frac{|\vec\alpha|}{\norm{\vec\alpha}_{1,\mathcal{P}}}.
\end{equation}
Protocols in our family satisfy this condition by construction as, for any valid protocol,
\begin{equation}
\sum_{m=1}^M |\vec\omega^{(m)}|=|W|\vec r,
\end{equation}
where $|W|$ denotes taking the element-wise absolute value of the elements of $W$. Consequently, noting that $\sgn(\omega^{(m)}_j)=\sgn(\alpha_j)$ for all $m$, we require
\begin{equation}
W\vec r=MN\frac{\vec\alpha}{\norm{\vec\alpha}_{1,\mathcal{P}}},
\end{equation}
which is Eq.~(\ref{eq:main_equation-2}).

\section{Proof of Lemma 1}\label{app:proof-of-lemma-1}

Here we provide a proof of Lemma~\ref{lemma:opt_states} in the main text, restated here for convenience.
\begin{lemma}
Any optimal protocol using $N$ photons and $M$ passes through interferometers with a coupling as in Eq.~(\ref{eq:initialH}) with $\hat{g}_j=\hat{n}_j$ requires that, for every pass $m$, the probe state $\ket{\psi_m}$ be of the form
\begin{equation}
    \ket{\psi_m} \propto \ket{\vec{N}(m)}_{\mathcal{P}}\ket{\vec{0}}_{\mathcal{NR}} + e^{i\varphi_m}\ket{\vec{0}}_{\mathcal{P}}\ket{\vec{N'}(m)}_{\mathcal{NR}},
\end{equation}
where $\mathcal{P}$, $\mathcal{N}$, and $\mathcal{R}$ represent the modes with $\alpha_j\geq 0$, $\alpha_j<0$, and the (arbitrary number of) reference modes, respectively, $\vec{N}(m)$ and $\vec{N'}(m)$ are strings of occupation numbers such that $|\vec{N}(m)|=|\vec{N'}(m)|=N$ for all passes $m$. $\varphi_m$ is an arbitrary phase. 
\end{lemma}
\begin{proof}
The quantum Fisher information matrix elements for any protocol with $\hat{g}_j=\hat{n}_j$ are given by
\begin{align}\label{eq:fim_general_app}
\mathcal{F}(\vec\theta)_{ij}&=4\Bigg[\sum_{l=1}^M\sum_{m=1}^M\frac{1}{2}\langle\{\hat{n}_i(l),\hat{n}_j(m)\}\rangle-\left(\sum_{m=1}^M\langle\hat{n}_i(m)\rangle\right)\left(\sum_{m=1}^M\langle\hat{n}_j(m)\rangle\right)\Bigg]\nonumber \\
&=4\sum_{l=1}^M\sum_{m=1}^M \mathrm{Cov}\left(\hat{n}_i(l),\hat{n}_j(m) \right),
\end{align}
where the expectation values are taken with respect to the initial probe state, and $\hat n_j(m)$ are the number operators on the $j^\mathrm{th}$ mode in the Heisenberg picture prior to the $m^\mathrm{th}$ pass, as specified in Eq.~(\ref{eq:generator}). Without loss of generality, we make the assumption that $\norm{\vec\alpha}_{1,\mathcal{P}}\geq \norm{\vec\alpha}_{1,\mathcal{N}}$. Summing over $i,j\in\mathcal{P}$, we have that, for an optimal protocol,
\begin{align}\label{eq:finalcond_res}
\sum_{i\in\mathcal{P}}\sum_{j\in\mathcal{P}}\mathcal{F}(\vec\theta)_{ij}=\sum_{j\in\mathcal{P}}\frac{(MN)^2}{\norm{\vec\alpha}_{1,\mathcal{P}}}\alpha_j=(MN)^2,
\end{align}
where we used the condition in Eq.~(\ref{eq:finalcond}) for an optimal protocol, and we recall that, for $j \in \mathcal{P}$, all $\alpha_{j} > 0$. For convenience, define 
\begin{equation}
\hat{P}(m):=\sum_{j\in\mathcal{P}}\hat{n}_j(m).
\end{equation}
Armed with this definition, we can upper bound the sum over $i,j\in\mathcal{P}$ in the explicit expression from Eq.~(\ref{eq:fim_general_app}) as

\begin{align}
\sum_{i\in\mathcal{P}}\sum_{j\in\mathcal{P}}\mathcal{F}(\vec\theta)_{ij}&=4\sum_{l=1}^M\sum_{m=1}^M \mathrm{Cov}\left(\hat{P}(l),\hat{P}(m) \right) \nonumber\\
&\leq 4\sum_{l=1}^M\sum_{m=1}^M \sqrt{\mathrm{Var}(\hat{P}(l))\mathrm{Var}(\hat{P}(m))}=4\left(\sum_{l=1}^M\sqrt{\mathrm{Var}(\hat{P}(l))}\right)^2\nonumber\\
&\leq 4\left(\sum_{l=1}^M\frac{\norm{\hat{P}(l)}_{s,N}}{2}\right)^2 \nonumber\\
&\leq (NM)^2, \label{eq:bnd_on_fim}
\end{align}
where in the first line we use the Cauchy-Schwarz inequality, in the second line we use that once restricted to the $N$-particle subspace $\mathrm{Var}(A)\leq \norm{A}_{s,N}^2/4$ (where, again, $\norm{A}_{s,N}$ is the seminorm restricted to the $N$-particle subspace) for any Hermitian operator $A$, and in the final line we use that $\norm{\hat{P}(l)}_{s,N}\leq N$. Comparing Eq.~(\ref{eq:bnd_on_fim}) with Eq.~(\ref{eq:finalcond_res}), we find that, for any optimal protocol, all inequalities in Eq.~(\ref{eq:bnd_on_fim}) must be saturated. Specifically,
\begin{align}
    \mathrm{Cov}\left(\hat{P}(l),\hat{P}(m) \right)^2 &= \mathrm{Var}(\hat{P}(l))\mathrm{Var}(\hat{P}(m)), \label{eq:saturation_condition_1}\\
    \mathrm{Var}(\hat{P}(l)) &= \frac{N^2}{4} \label{eq:saturation_condition_2}.
\end{align}
The second condition, Eq.~(\ref{eq:saturation_condition_2}), means that, at all times, the state of our system must be of the form
\begin{equation}
    \frac{\ket{\vec{N}(l)}_{\mathcal{P}}\ket{\vec{0}}_{\mathcal{NR}} + e^{i\varphi_l}\ket{\vec{0}}_{\mathcal{P}}\ket{\vec{N'}(l)}_{\mathcal{NR}}}{\sqrt{2}},
\end{equation}
where we are using the simplifying notation from the statement of the lemma. In particular, the subscripts $\mathcal{P}, \mathcal{N}, \mathcal{R}$ refer to the collection of all modes associated with $\alpha_j \geq 0, \alpha_j < 0$, and the reference modes, respectively. Therefore, the state $\ket{\vec{N}}_{\mathcal{P}}\ket{\vec{0}}_{\mathcal{NR}}$ means that all photons are distributed (in some potentially arbitrary way) amongst the modes with non-negative $\alpha_j$, and there are no photons in the modes with negative $\alpha_j$ or in the reference modes. Contrastingly, $\ket{\vec{0}}_{\mathcal{P}}\ket{\vec{N'}(l)}_{\mathcal{NR}}$ refers to a state where there are $N$ photons in the negative and reference modes, and there are no photons in the non-negative modes. We have also shifted to the Schr{\"o}dinger picture where we move the time dependence onto the state as opposed to the operators. It is simple to verify that this state satisfies Eq.~(\ref{eq:saturation_condition_2}), and it is also simple to verify these are the most general states that achieve this. Intuitively, $\ket{\psi_m}$ is a generalized N00N state between the positive and negative/reference modes.
\end{proof}

In addition, we have the following useful characterization of optimal protocols:
\begin{lemma}\label{lemma:coherent-mapping}
Let $\ket{\psi_{i}}$ be a state of the form in Lemma~\ref{lemma:opt_states}. Refer to the first and second parts of its superposition as, respectively, the first and second or positive and non-positive branches. Let $U_{m}$ be the unitary that maps the initial state $\ket{\psi_{1}}$ to the state just before the $m$-th pass, $\ket{\psi_{m}}$, given by
\begin{equation}
    U_m = \begin{cases}
    \prod_{i=1}^{m-1} U^{(i)}V, & M+1\geq m\geq 2 \\
    I, & m=1.
    \end{cases}
\end{equation}
in agreement with Eq.~(\ref{eq:full_unitary}). Then, if $U_{m}$ is part of an optimal protocol, it coherently maps the first (second) branch of $\ket{\psi_1}$ to the first (second) branch of $\ket{\psi_m}.$
\end{lemma}
\begin{proof}
We use the covariance equality in Eq.~(\ref{eq:saturation_condition_1}).  To proceed, we evaluate the expectation value of $\hat{P}$ in the initial state. Here, we will again use the Schr{\"o}dinger picture. 
\begin{align}
    \bra{\psi_1}\hat{P}(l)\ket{\psi_1} &= \bra{\psi_{l}}\hat{P}\ket{\psi_{l}} \\
    &= \frac{1}{2}\left( \bra{\vec{N}(l)}_{\mathcal{P}}\bra{\vec{0}}_{\mathcal{NR}} + e^{-i\varphi_l}\bra{\vec{0}}_{\mathcal{P}}\bra{\vec{N'}(l)}_{\mathcal{NR}}\right)\hat{P}\left(\ket{\vec{N}(l)}_{\mathcal{P}}\ket{\vec{0}}_{\mathcal{NR}} + e^{i\varphi_{l}}\ket{\vec{0}}_{\mathcal{P}}\ket{\vec{N'}(l)}_{\mathcal{NR}}\right) \\
    &=\frac{1}{2}\left( \bra{\vec{N}(l)}_{\mathcal{P}}\bra{\vec{0}}_{\mathcal{NR}} + e^{-i\varphi_l}\bra{\vec{0}}_{\mathcal{P}}\bra{\vec{N'}(l)}_{\mathcal{NR}}\right)N\left(\ket{\vec{N}(l)}_{\mathcal{P}}\ket{\vec{0}}_{\mathcal{NR}}\right)\\
    &= \frac{N}{2}. 
\end{align}
We next evaluate the covariance:
\begin{align}
    \mathrm{Cov}\left(\hat{P}(l), \hat{P}(m)\right) &= \bra{\psi_1}\hat{P}(l)\hat{P}(m)\ket{\psi_1} - \bra{\psi_1}\hat{P}(l)\ket{\psi_1}\bra{\psi_1}\hat{P}(m)\ket{\psi_1} \\
    &= \bra{\psi_l}\hat{P}U_{l} U^{\dag}_{m} \hat{P}\ket{\psi_m} - \bra{\psi_l}\hat{P}\ket{\psi_l}\bra{\psi_m}\hat{P}\ket{\psi_m} \\
    &= \frac{N^2}{2}\bra{\vec{N}(l)}_{\mathcal{P}}\bra{\vec{0}}_{\mathcal{NR}}U_{l} U^\dag_{m} \ket{\vec{N}(m)}_{\mathcal{P}}\ket{\vec{0}}_{\mathcal{NR}} - \frac{N^2}{4},
\end{align}
where in the last line we have used the fact that $\hat{P}$ gives a factor of $N$ when acting on the first branch of states $\ket{\psi_l}$ and $\ket{\psi_m}$, but it annihilates the second branch that has zero photons in the positive modes. 

In order for Eq.~(\ref{eq:saturation_condition_1}) to be satisfied, and using Eq.~(\ref{eq:saturation_condition_2}), we therefore require that, for all pairs of passes $l, m,$
\begin{equation}
    \bra{\vec{N}(l)}_{\mathcal{P}}\bra{\vec{0}}_{\mathcal{NR}}U_{l} U^\dag_{m} \ket{\vec{N}(m)}_{\mathcal{P}}\ket{\vec{0}}_{\mathcal{NR}} = 1.
\end{equation}
Choosing $l = 1$, this implies that we require that
\begin{equation}
    U^\dag_m \ket{\vec{N}(m)}_{\mathcal{P}}\ket{\vec{0}}_{\mathcal{NR}} = \ket{\vec{N}(0)}_{\mathcal{P}}\ket{\vec{0}}_{\mathcal{NR}} \equiv \ket{\psi_1}_{\mathcal{P}},
\end{equation}
where we are defining $ \ket{\psi_1}_{\mathcal{P}} ,\ket{\psi_1}_{\mathcal{NR}}$ such that
$\ket{\psi_0} \propto \ket{\psi_1}_{\mathcal{P}} + \ket{\psi_1}_{\mathcal{NR}}$ in the obvious way. Moving the unitary onto the right hand side of the equation yields
\begin{equation}
    \ket{\psi_m}_{\mathcal{P}} =  U_m\ket{\psi_1}_{\mathcal{P}},
\end{equation}
which of course implies the corresponding equation for the second branch by linearity.
\end{proof}

\section{Fisher Information Matrix Conditions for Quadrature Displacements}\label{app:fim_quad}
In this appendix, we provide conditions on the quantum Fisher information matrix for an optimal protocol in the case of quadrature generators. This result yields a simpler form of the saturability condition of Eq.~(\ref{eq:finalcondquad}), although the set of states that it picks out is less clear than in the number operator case. This issue is compounded by the fact that the bound is not actually saturable (it can only be approached asymptotically as $\overline{N}\rightarrow\infty$). Regardless, it allows us to bring quadrature displacements into our general formalism and suggests a route towards designing additional optimal protocols beyond those already in the literature. 

In particular, starting with the definition of $\hat p_i(l)$ from Eq.~(\ref{eq:fim_general_quad}), we can bound the sum over the quantum Fisher information matrix elements as
\begin{align}
\sum_{i=1,j=1}^{d} \mathcal{F}(\vec\theta)_{ij} &= \sum_{i=1,j=1}^{d}4\sum_{l=1}^M\sum_{m=1}^M \mathrm{Cov}(\hat p_i(l), \hat p_j(m)) \\
&\leq 4\sum_{l=1}^M\sum_{m=1}^M\sqrt{\mathrm{Var}\Big(\sum_{i=1}^{d} \hat p_i(l)\Big)\mathrm{Var}\Big(\sum_{i=1}^{d} \hat p_j(m)\Big)} \label{eq:aa}\\
&= 4 \left(\sum_{l=1}^M \sqrt{\mathrm{Var}\Big(\sum_{i=1}^{d} \hat p_i(l)\Big)}\right)^2 \\
&\leq 4 \left(\sum_{l=1}^M\sqrt{\overline{N}-\frac{d}{2}} \right)^2 = 4M^2 \left(\overline{N}-\frac{d}{2}\right) \sim 4M^2\overline{N}. \label{eq:aaa}
\end{align}
Above, in Eq.~(\ref{eq:aa}), we used the Cauchy-Schwarz inequality; in Eq.~(\ref{eq:aaa}), we used the uncertainty relation in Eq.~(\ref{eq:a}). Consistent with the rest of the paper, the $\sim$ symbol denotes asymptotically in $\overline{N}$ (for $\overline{N}\gg d$).

The saturability condition in Eq.~(\ref{eq:finalcondquad}) states that, for an optimal protocol (asymptotically in $\overline{N}$), it must hold that $\vec\alpha$ is an eigenvector of $\mathcal{F}(\vec\theta)$ with eigenvalue $4M^2\overline{N}$.  Thus, for an optimal protocol,
\begin{align}
\mathrm{Tr}(\mathcal{F})=\sum_{j=1}^d \lambda_j \gtrsim 4M^2\overline{N},
\end{align}
where $\lambda_j$ are the eigenvalues of $\mathcal{F}$. This implies that the chain of inequalities leading to Eq.~(\ref{eq:aaa}) must be saturated (asymptotically in $\overline{N}$) for an optimal protocol and that the largest eigenvalue of $\mathcal{F}$ must be $\lambda_1\sim 4m^2\overline{N}$ with all other eigenvalues zero. It immediately follows that the saturability condition  for quadrature displacements can be written as
\begin{equation}
\mathcal{F}(\vec\theta)_{ij}\sim \frac{4M^2\overline{N}}{\norm{\vec\alpha}_2^2}\alpha_i\alpha_j.
\end{equation}

\section{Approaching the Single-Shot Limit and Robust Phase Estimation}\label{app:robust-phase-estimation}

As pointed out in the footnote preceding \cref{eq:qcrb2} and in the discussion of what defines an information-theoretically optimal protocol in Sec.~\ref{s:new-prot}, it is not, in practice, possible to construct an unbiased estimator achieving the single shot ($\mu=1$) quantum Cram\'{e}r-Rao bound that we analyze in this paper, as the quantum Cram\'{e}r-Rao bound is only guaranteed to be achievable in the limit of asymptotically large amounts of data ($\mu\rightarrow\infty$). Resolving this tension while still achieving asymptotic Heisenberg scaling in the total amount of resources (here, $\mu N$ photons) requires carefully designed protocols. In particular, extracting a relative phase from the probe states considered in the protocols in this paper requires a proper division of resources so that, asymptotically, the single-shot bound is achieved up to a small constant. 

At best, this constant can be reduced to $\pi^2$~\cite{gorecki2020pi}, but the non-adaptive robust phase estimation scheme of Refs.~\cite{kimmel2015robust,kimmel2015robusterratum,belliardo2020achieving} provides a relatively simple-to-implement approach with a multiplicative overhead of $(24.26\pi)^2$. In brief, these protocols work by dividing the protocol into $K$ stages where in stage $j$ one uses $N_j$ photons (or $\overline{N}_{j}$ average photons for displacement sensing). In each stage, one imprints the unknown function into the phase between two branches of a cat-like state of $N_{j}$ photons and then performs a measurement, as described in the main text. The experiment is performed $\nu_j$ times, allowing one to obtain an estimate of the unknown phase. This estimate is refined over the course of the $K$ stages, with more photons used in each additional stage such that the total photon resources are 
\begin{equation}
N=\sum_{j=1}^K \nu_j N_j.
\end{equation}

An optimal choice of $\nu_j$ and $N_j$ ensures that, asymptotically, $N_K=\Theta(N)$ and $\nu_K=\mathcal{O}(1)$, and, thus, the asymptotic scaling of the single-shot bound is obtained up to a multiplicative constant that depends on the details of the optimization. The proof of this and the associated optimization are detailed in Refs.~\cite{kimmel2015robust,kimmel2015robusterratum,belliardo2020achieving}.

\end{appendix}

\bibliography{main.bib}
\end{document}